\documentclass[11pt,a4paper]{article}
\usepackage{amsmath, amsfonts, amssymb}
\usepackage{a4wide}
\usepackage[left=0.8in,bottom=1in,top=1in,right=0.8in]{geometry}
\usepackage{parskip}
\usepackage{enumitem}
\usepackage{xcolor}

\usepackage{hyperref}
\usepackage{soul} 
\usepackage{booktabs}
\usepackage{caption}
\usepackage{siunitx}
\usepackage{tabularx}
\usepackage{comment}
\usepackage{graphicx}
\usepackage{adjustbox}
\usepackage{multirow,tensor}
\usepackage{amsthm}
\usepackage{bm, makecell}
\usepackage{lscape}
\usepackage{rotating}
\usepackage{floatrow}
\usepackage{threeparttable} 
\usepackage[normalem]{ulem}
\usepackage[noadjust]{cite}

\def\qed{\hfill {$\square$}\goodbreak \medskip}

\newtheorem{theorem}{Theorem}[section]
\newtheorem{lemma}[theorem]{Lemma}
\newtheorem{corollary}[theorem]{Corollary}

\newtheorem{definition}[theorem]{Definition}
\newtheorem{example}[theorem]{Example}
\newtheorem{remark}[theorem]{Remark}

\numberwithin{equation}{section}

\newcommand{\Supp}{\textnormal{Supp}}
\DeclareMathOperator{\wt}{wt}

\usepackage{tikz,xcolor,hyperref}
\usepackage{mathdots}

\definecolor{lime}{HTML}{A6CE39}
\DeclareRobustCommand{\orcidicon}{%
	\begin{tikzpicture}
		\draw[lime, fill=lime] (0,0) 
		circle [radius=0.16] 
		node[white] {{\fontfamily{qag}\selectfont \tiny ID}};
		\draw[white, fill=white] (-0.0625,0.095) 
		circle [radius=0.007];
	\end{tikzpicture}
	\hspace{-2mm}
}

\foreach \x in {A, ..., Z}{%
	\expandafter\xdef\csname orcid\x\endcsname{\noexpand\href{https://orcid.org/\csname orcidauthor\x\endcsname}{\noexpand\orcidicon}}
}

\begin{document}
	\date{}
		\title{New Quaternary codes with small Plotkin-defects from two-generator simplicial complexes}

		\author{{\bf Ankit Yadav\footnote{email: {\tt ankityadav10102000@gmail.com}}\orcidA{},\; 
        \bf Nilay Kumar Mondal\footnote{email: {\tt nilaym@srmist.edu.in }}\orcidB{} and \bf
        Ritumoni Sarma\footnote{email: {\tt ritumoni407@gmail.com}}\orcidC{}} \\ $^{\ast \ddagger }$Department of Mathematics\\ Indian Institute of Technology Delhi\\Hauz Khas, New Delhi-110016, India\\
        $^{\dagger }$Department of Mathematics\\ SRM Institute of Science and Technology\\Kattankulathur,
Tamil Nadu 603203, India
        \medskip \\
  }
  
\maketitle

\begin{abstract}
A recent characterization of all lengths of Plotkin-optimal quaternary (that is, over the ring $\mathbb{Z}_4$) codes of arbitrary type \cite{tang2025plotkin} also pins down the parameters for which no Plotkin-optimal code exists. In this article, we determine the best achievable parameters in several of these cases, obtaining codes whose minimum Lee distance is one less than the Plotkin bound, namely the codes with Plotkin-defect 1. To the best of our knowledge, this is the first attempt to study quaternary codes with Plotkin-defects. Precisely, we construct infinite families of quaternary  $\mathcal{C}_{D}$-codes, where the defining set $D$ is derived utilizing a two-generator simplicial complex, and determine their Lee weight distributions. As a result, we find two quaternary linear code families with Plotkin-defect 1 and report at least 23 new or improved parameters having small (upto 4) Plotkin-defects, including 13 projective and 7 optimal parameters. We additionally report 4 quaternary linear codes with best-known parameters that are also projective. Further,  their linear Gray images give two infinite families of distance-optimal, one infinite family of at least almost dimension-optimal binary linear codes and five infinite families of minimal binary linear codes.

\medskip

			
\medskip
			

\end{abstract}

\section{Introduction}\label{Sec1}
\subsection{Background}
The study of codes over rings has attracted significant attention following the remarkable article by Hammons et al. \cite{hammons1994Z4}, which showed that several well-known binary non-linear codes can be realized as the Gray images of $\mathbb{Z}_4$-linear codes. The term quaternary is used in the literature for both codes over the finite field $\mathbb{F}_{4}$ and codes over $\mathbb{Z}_{4}$. In this article, however, codes over $\mathbb{Z}_{4}$ are referred to as quaternary codes. Quaternary linear codes exhibit deep connections with lattices, combinatorial designs, and low-correlation sequences, as demonstrated in \cite{wan1997quaternary}. Since then, the construction of quaternary linear codes with good parameters has become an important problem in coding theory, and various works have been done in this direction \cite{shi2017consta, Meng2024generalized,habibul2021z4,maheshanand2017skew}. We also direct the reader to several recent studies on this topic \cite{zhu2019new,wang2026trace,hyun2025multivariable
,shi2019trace,carlet2000one,dougherty2001maximum,Kiermaier2016new,Kiermaier2011hyperoval,Kiermaier2013new,shi2014optimal,tang2025plotkin,tang2023new}.

On the other hand, simplicial complexes have been shown to be a powerful tool for constructing projective few-weight binary linear codes, as evidenced by the following literature. Chang and Hyun \cite{Chang_Hyun2018simplicial} introduced a construction of binary linear codes from simplicial complexes and obtained infinite families of optimal few-weight binary linear codes. Subsequently, in \cite{hyun2020infinite}, the authors constructed infinite families of optimal few-weight binary linear codes from one or two-generator simplicial complexes by employing the defining-set technique proposed by Ding and Niederreiter \cite{ding2007cyclotomic}. In \cite{wu2020binary}, the authors used punctured one-generator simplicial complexes to construct few-weight binary linear codes which are distance-optimal, self-orthogonal or LCD. Thereafter, many researchers have used simplicial complexes to construct several families of optimal few-weight binary linear codes via different approaches. One can refer to \cite{yadav2025optimal, vidya2024nonunital, Mondal2024mixed_alphabet, wu2020optimal, hu2024new, chen2025optimal, wu2024survey, shi2022few-weight, anuj2025subfield, chen2025griesmer} and the references therein for a more detailed study. Furthermore, simplicial complexes have also been used as an effective tool for constructing projective minimal binary linear codes. For example, the authors in \cite{Chang_Hyun2018simplicial} obtained an infinite family of minimal binary linear codes that do not satisfy the Ashikhmin-Barg condition. The construction of minimal linear codes has attracted considerable research interest due to their wide applications in secret sharing schemes \cite{shamir1979secret, yuan2006secret} and secure two-party computation \cite{chabanne2014twoparty}. In addition, these codes are amenable to minimum-distance decoding \cite{Ashikhmin1998}, which further enhances their theoretical and practical significance.

\subsection{Motivation and Objectives}
Wu et al. \cite{wu2024quaternary} first introduced the use of binary simplicial complexes to construct two infinite families of quaternary linear codes from two-generator simplicial complexes and determined their Lee weight distributions. They reported at least 9 new quaternary linear codes as per the database \cite{aydin2022updated} and also obtained infinite families of minimal binary linear codes as Gray images. However, all the resulting codes are Lee weight non-projective in contrast to the fact that simplicial complexes are utilized in the binary case to yield projective codes. Further, in \cite{wu2024two}, Wu et al. employed one-generator simplicial complexes as the support set of Boolean functions to define two infinite families (one linear and one non-linear) of quaternary codes. For other studies on quaternary linear codes constructed from defining sets, we refer the reader to \cite{hyun2025multivariable,shi2019trace,wang2026trace,zhu2019new}. Very recently, Tang \cite{tang2025plotkin} completely characterized all possible lengths of Plotkin-optimal codes of arbitrary type, together with their corresponding existence results. A particularly significant consequence of this work is that it yields explicit parameter sets for which no Plotkin-optimal code can exist.

Motivated by the above developments on quaternary linear codes, especially by employing simplicial complexes and nonexistence results of Plotkin-optimal codes for several  parameters, we raise the following questions:
\begin{itemize}
    \item Can good quaternary linear codes be constructed using a defining set consisting of punctured two-generator simplicial complexes, a case that, to the best of our knowledge, has not yet been explored in the quaternary setting? 

    \item If so, can the defining set be further modified to construct quaternary linear codes that are also Lee weight projective?  Subsequently, is it possible to find even some projective quaternary linear codes with best-known parameters that outperform the currently reported best-known codes due to their projectivity?

    \item In \cite{wu2024quaternary}, defining sets were considered where a two-generator and a one-generator simplicial complex share a common generator. Can the Lee weight distribution be determined when such a restriction is removed, and does this lead to good quaternary linear codes with more flexible lengths?

    \item  Can these constructions lead us to the best possible codes when Plotkin-optimality is unattainable, more precisely, codes that fall just short of the Plotkin bound, namely, codes with ``Plotkin-defect" 1?
\end{itemize}

\subsection{Our contribution}

First, we noticed that the main Lee weight formula used in \cite{wu2024quaternary} to calculate the Lee weight distribution of the quaternary linear codes constructed in their paper is not correct in general. We explain this in Remark \ref{correction} and rectify it by adding a condition as given in Equation \ref{corrected} in Section \ref{Sec2}. Now, to answer the above questions, we use a punctured two-generator simplicial complex, namely, $\Delta_{B,C}\setminus \Delta_{B \cap C}$ to employ the defining set $D=\Delta_{A}+ 2 (\Delta_{B,C}\setminus \Delta_{B \cap C})$ with $|A|=1$ in order to construct infinite families of quaternary linear codes $\mathcal{C}_{D}$. Further, we modified this defining set as $D=\Delta_{A}^{*}+ 2 (\Delta_{B,C}\setminus \Delta_{B \cap C})$ with $|A|=1$ and $A\cap(B\cup C) = \emptyset$ to construct infinite families of Lee weight projective quaternary linear codes. Lastly, we use the complement of a two-generator simplicial complex $(\Delta_{B,C})^{c}$ to employ the defining set $D=\Delta_{A}+ 2 (\Delta_{B,C})^{c}$ with $|A|=1$. We must note that in all three cases above, the generators are distinct, which removes the restriction on the set of generators of the simplicial complexes used as defining sets in \cite{wu2024quaternary}. As an immediate advantage, this enables the construction of quaternary codes with odd lengths, which was not achievable in \cite{wu2024quaternary}. 
\vskip 2pt
Further, we find the Lee weight distributions for all three choices of defining sets and obtain several quaternary linear codes with small (upto 4) Plotkin-defects, and report at least 23 new or improved quaternary linear codes by comparing with the database of best-known quaternary codes \cite{aydin2022updated} and with all possible parameters of Plotkin-optimal codes \cite{tang2025plotkin}. Among these 23 codes, 13 are Lee weight projective and 7 are optimal. We additionally report 4 projective quaternary linear codes with best-known parameters that might outperform the currently reported best-known codes due to their projectivity. Moreover, we obtain two families of quaternary linear codes with Plotkin-defect 1, that is, the codes in this family are either optimal or almost optimal, depending on whether Plotkin-optimal codes of the corresponding length and type exist. For the convenience of readers, we clarify the following terminologies: We use `new' to describe a particular parameter set of a quaternary linear code for which no Plotkin-optimal code of that length and type is reported in \cite{tang2025plotkin}, and no code of the same length and type is reported in the database \cite{aydin2022updated}. We use `improved' when no Plotkin-optimal code of that length and type is reported in \cite{tang2025plotkin}, but the best-known quaternary code of the same length and type in the database has a smaller Lee distance than the code we construct. By `projective' quaternary linear codes, we mean Lee weight projective, and `family' refers to an infinite family.
\vskip 2pt
Furthermore, we prove that their Gray images are linear. By analyzing the parameters and weight distribution of the linear Gray image, we obtain a family of Griesmer codes already reported in \cite[Theorem 5]{chen2025griesmer}, a distance-optimal family and at least one almost dimension-optimal family, along with several families of binary minimal codes. Moreover, by applying the MacWilliams identities to the Gray image parameters of a family of quaternary linear codes constructed in this paper, we found 9 new parameters from the corresponding dual codes, presented in Table \ref{table6}.

The article is organized as follows. The subsequent section presents the necessary preliminaries. Section \ref{Sec3} comprises the Lee weight distributions of the three families of quaternary codes. In Section \ref{Sec4}, we characterize the conditions for the linearity and minimality of the quaternary codes obtained in Section \ref{Sec3}. Finally, we conclude our article in Section \ref{Sec5}.

\section{Preliminaries} \label{Sec2}
Throughout this article, we use the following notation.
\begin{table}[H]
\centering
\begin{tabular}{l|l}
\hline
Notation & Description \\
\hline
$m$ & a positive integer \\
$[m]$ & the set $\{1, 2, \dots, m\}$\\
$\#S$ &  the cardinality of $S$ \\
$\mathcal{P}(S)$ & the power set of $S$\\
$\mathbb{Z}_{2}$ & the ring of integers modulo $2$ \\
$\Supp(\mathbf{v})$ & $\{j\in [m]: v_j\ne 0\}$, for $\mathbf{v}\in\mathbb{Z}_{2}^{m}$\\
$\wt_{H}(\mathbf{v})$ & $\#\Supp (\mathbf{v})$\\
$\mathbf{x}\mathbf{y}$ & the standard Euclidean inner product of $\mathbf{x},\mathbf{y}\in \mathbb{Z}_{2}^{m}$\\
\hline
\end{tabular}
\end{table}
A \textit{ binary linear code} $\mathcal{C}$ with parameters $[n,k,d]$ of length $n$ over $\mathbb{Z}_{2}$ is a $k$-dimensional subspace of $\mathbb{Z}_{2}^{n}$, 
where $d$ is the minimum (Hamming) distance of $\mathcal{C}$. Let $A_{j}$ be the number of codewords of $\mathcal{C}$ having weight $j$, then the string $(A_{0},A_{1},\ldots,A_{n})$ is said to be the \textit{Hamming weight distribution} of $\mathcal{C}$. The code $\mathcal{C}$  is called a $t$-\textit{weight} linear code if exactly $t$ entries among $A_{1},A_{2},\ldots,A_{n}$ are non-zero. 
Let $\mathcal{C}^{\perp}$ be the dual code of $\mathcal{C}$ with respect to the usual inner product. The MacWilliams identities relate the weight distributions of $\mathcal{C}$ and $\mathcal{C}^{\perp}$, and are given as follows.
\begin{lemma} \cite{huffman2010book}\label{lem:macwilliams}
    Let $\mathcal{C}$ be an $[n,k,d]$ binary linear code with weight distribution $(A_{0},A_{1},\ldots,A_{n})$, and let the weight distribution of $\mathcal{C}^{\perp}$ be $(A_{0}^{\perp},A_{1}^{\perp},\ldots,A_{n}^{\perp})$. Then 
\begin{equation}\label{eq:MacWilliams}
\sum_{j=0}^{n-r} \binom{n-j}{r} A_j
=
2^{k-r}
\sum_{j=0}^{r} \binom{n-j}{n-r} A_j^{\perp}
\qquad \text{for } 0 \le r \le n.
\end{equation}
\end{lemma}

\subsection{Quaternary codes}\label{subsection 2.1}
Let $\mathbb{Z}_{4}$ be the ring of integers modulo $4$. Any non-empty subset $\mathcal{C}$ of $\mathbb{Z}_{4}^{n}$ is called a \textit{quaternary code}. A \textit{linear code} $\mathcal{C}$ of length $n$ over $\mathbb{Z}_{4}$, also called a \textit{quaternary linear code} or a $\mathbb{Z}_{4}$-\textit{linear code}, is a $\mathbb{Z}_{4}$-submodule of $\mathbb{Z}_{4}^{n}.$ 

Let $\mathcal{C}$ and $\mathcal{D}$ be two $\mathbb{Z}_{4}$-linear codes of length $n$. The codes $\mathcal{C}$ and $\mathcal{D}$ are said to be \textit{permutation equivalent} if there exists an $n\times n$ permutation matrix $P$ such that $\mathcal{D}=\mathcal{C}P.$

Any $\mathbb{Z}_{4}$-linear code $\mathcal{C}$ is permutation equivalent to a code with generator matrix $G$ of the form
\[
G = \left[
\begin{tabular}{c c c}
    $I_{k_{1}}$ & $A_{1}$ & $B_{1}+2B_{2}$  \\
     $O$ & $2I_{k_{2}}$ & $2A_{2}$ 
\end{tabular}
\right],
\]
where $A_{1},A_{2},B_{1}$, and $B_{2}$ are matrices over $\mathbb{Z}_{2}$, and $O$ denotes the $k_{2}\times k_{1}$ zero matrix. The code $\mathcal{C}$ is said to be of type $4^{k_{1}}2^{k_{2}}.$

We recall the Gray map \cite{hammons1994Z4}, $\phi:\mathbb{Z}_{4}\to \mathbb{Z}_{2}^{2}$, defined by
\[
\phi(0) = 00, \phi(1) = 01, \phi(2) = 11, \text{ and } \phi(3) = 10.
\]
Since every element of $\mathbb{Z}_{4}$ admits a unique representation of the form $b+2c$, where $b,c \in \mathbb{Z}_{2}$, the Gray map $\phi$ can be equivalently expressed as 
\[
b+2c \mapsto (c, b+c).
\]
This map extends componentwise to a map, $\Phi:\mathbb{Z}_{4}^{n} \to \mathbb{Z}_{2}^{2n}$, defined by
\[
\mathbf{b}+2\mathbf{c} \mapsto (\mathbf{c}, \mathbf{b}+\mathbf{c}),
\]
where $\mathbf{b},\mathbf{c}\in\mathbb{Z}_{2}^{n}.$

The\textit{ Lee weight} of a vector $\mathbf{v} = \mathbf{b}+2\mathbf{c}\in\mathbb{Z}_{4}^{n}$ is defined as 
\[
\wt_{L}(\mathbf{v}) := \wt_{H}(\Phi(\mathbf{v})) = \wt_{H}(\mathbf{c})+\wt_{H}(\mathbf{b}+\mathbf{c}).
\]
The \textit{minimum Lee weight} of $\mathcal{C}$ (denoted by $\wt_{L}(\mathcal{C})$) is defined as 
\[
\wt_{L}(\mathcal{C}) := \min\{ \wt_{L}(\mathbf{c}) \mid \mathbf{c} \in \mathcal{C}, \mathbf{c}\neq \mathbf{0}\}.
\]
Suppose $\mathcal{C}$ is a $\mathbb{Z}_{4}$-linear code of length $n$.
For $0\leq i \leq 2n,$ let $A_{i}$ denotes the number of codewords of $\mathcal{C}$ having Lee weight $i$. The sequence $(A_{0}, A_{1}, \ldots, A_{2n})$ is said to be the \textit{Lee weight distribution} of $\mathcal{C},$ and the homogeneous polynomial 
\[
\text{Lee}_{\mathcal{C}}(x,y) = \sum\limits_{\mathbf{c}\in\mathcal{C}} x^{\wt_{L}(\mathbf{c})}y^{2n-\wt_{L}(\mathbf{c})}
\]
is said to be the \textit{Lee weight enumerator} of $\mathcal{C}$. 

The \textit{Lee distance} between two vectors $\mathbf{u}$ and $\mathbf{v}$ is defined as 
\[
d_{L}(\mathbf{u},\mathbf{v}) := \wt_{L}(\mathbf{u}-\mathbf{v}).
\]
The \textit{minimum Lee distance} of $\mathcal{C}$ (denoted by $d_{L}(\mathcal{C})$) is defined as 
\[
d_{L}(\mathcal{C}) := \min\{ d_L(\mathbf{x},\mathbf{y}) \mid \mathbf{x},\mathbf{y} \in \mathcal{C}, \mathbf{x}\neq \mathbf{y}\}.
\]
If $\mathcal{C}$ is a $\mathbb{Z}_{4}$-linear code, then its minimum Lee weight coincides with its minimum Lee distance. We denote a quaternary code of length $n$, size $k$ and minimum Lee distance $d_{L}$ by $(n, k, d_{L})$.

The Gray map $\Phi$ is a distance preserving map from $\mathbb{Z}_{4}^{n}$, equipped with the Lee distance, to $\mathbb{Z}_{2}^{2n}$, equipped with the Hamming distance. Consequently, the Lee weight distribution of a $\mathbb{Z}_{4}$-linear code $\mathcal{C}$ coincides with the Hamming weight distribution of its Gray image $\Phi(\mathcal{C})$.

A linear code $\mathcal{C}$ of length $n$ over $\mathbb{Z}_{4}$ is called \textit{projective} if $d_{L}(\mathcal{C}^{\perp})\geq 3$, where $\mathcal{C}^{\perp}$ denotes the dual of $\mathcal{C}$ with respect to the standard Euclidean inner product. The dual code $\mathcal{C}^{\perp}$ is also a quaternary linear code of length $n$ and of type $4^{n-k_1-k_2}2^{k_2}$, provided $\mathcal{C}$ is of length $n$ and type $4^{k_1}2^{k_2}$ \cite{huffman2010book}.

The following lemma gives a necessary and sufficient condition for a quaternary linear code to be projective.
\begin{lemma}\cite[Lemma 4.8]{shi2020two}\label{lem 2.2}
	Let $\mathcal{C}$ be a length $n$ quaternary linear code of type $4^{k_1}2^{k_2}$ with a generator
	matrix $G$. Then $\mathcal{C}$ is projective if and only if the following two conditions hold:
    \begin{enumerate}
        \item every column of $G$ contains $1$ or $3$,
        \item no two columns of $G$ are $\pm1$ multiple of each other.
    \end{enumerate}
\end{lemma}
We next recall an upper bound on the minimum Lee distance of a quaternary linear code.
\begin{lemma}[Plotkin-type Lee distance bound \cite{Wyner1968upperbound, tang2023general}]\label{lem 2.1} 
	Let $\mathcal{C}$ be an $(n,k,d_L)$-quaternary linear code. Then $d_L\le\frac{kn}{k-1}$.
\end{lemma}
A quaternary linear code $\mathcal{C}$ with parameters $(n,k,d_L)$ is \textit{Plotkin-optimal} if $d_L=\lfloor\frac{kn}{k-1}\rfloor$. We define the Plotkin-defect of $\mathcal{C}$ as $\lfloor\frac{kn}{k-1}\rfloor-d_L$. We must note that, if there exists no Plotkin-optimal code with length $n$ and size $k$, a $(n,k,d_L)$ quaternary linear code having a Plotkin-defect $1$ is optimal, and a Plotkin-defect $2$ is at least almost optimal. Here, optimal refers to distance-optimal.

The following lemma characterizes the linearity of the Gray image of a quaternary linear code.
\begin{lemma}\cite[Theorem 12.2.3]{huffman2010book}\label{lem 2.3}
	Let $\mathcal{C}$ be a quaternary linear code of length $n$. Then the binary code $\Phi(\mathcal{C})$ is linear if and only if  $2({\bf{c}*\bf{d}})\in \mathcal{C}$ for all ${\bf{c},\bf{d}}\in \mathcal{C}$, where ${\bf{c}*\bf{d}}=(c_1d_1,\dots,c_nd_n)$ is the componentwise product in $\mathbb{Z}_4^n$.
\end{lemma}

\subsection{Simplicial Complex}
Let $\mathbf{v},\mathbf{w} \in \mathbb{Z}_{2}^{m}$. Then, $\mathbf{w}$ \textit{covers} $\mathbf{v}$ (we write $\mathbf{v} \preceq \mathbf{w}$) if $\Supp(\mathbf{w}) \supseteq \Supp(\mathbf{v})$. The map $\phi : \mathbb{Z}_{2}^{m} \to \mathcal{P}([m])$ given by $\mathbf{v}\mapsto \Supp(\mathbf{v})$ is a bijective map.
\begin{definition}
    A \textit{simplicial complex} $\Delta$ is a subset of $ \mathbb{Z}_{2}^{m}$ that satisfies the condition that $\mathbf{v}\in \Delta \implies \mathbf{w}\in \Delta$ for all $\mathbf{w}\preceq \mathbf{v}$. An element $\mathbf{v}\in\Delta$ is said to be \textit{maximal} if, whenever $\mathbf{v}\preceq\mathbf{w}$ for some $\mathbf{w}\in\Delta$, it follows that $\mathbf{v}=\mathbf{w}$.
\end{definition}
A simplicial complex generated by $A\subseteq [m]$ (denoted by $\Delta_{A}$) is a $|A|$-dimensional subspace of $\mathbb{Z}_{2}^{m}$ having one maximal element $A$, that is, $\Delta_{A} := \{\mathbf{v}\in\mathbb{Z}_{2}^{m}: \Supp(\mathbf{v}) \subseteq A\}.$ We write $\Delta_{A}^{*} = \Delta_{A}\setminus \{\mathbf{0}\}$. Furthermore, a simplicial complex generated by two maximal elements $A$ and $B$ (denoted by $\Delta_{A,B}$) is $\Delta_{A,B} := \Delta_{A} \cup \Delta_{B}$. Note that $ \Delta_{A} \cap \Delta_{B} = \Delta_{A\cap B}.$ We refer to any simplicial complex generated by two maximal elements as a \textit{two-generator simplicial complex}.

Let $\Delta_{A}\subseteq \mathbb{Z}_{2}^{m}$ be a simplicial complex. We denote by $\Delta_{A}^{\perp}$ the dual of $\Delta_{A}$ with respect to the standard Euclidean inner product. It is straightforward to verify that $\Delta_{A}^{\perp}= \{\mathbf{v}\in\mathbb{Z}_{2}^{m}: \Supp(\mathbf{v})\cap A = \emptyset\} = \Delta_{A^c}.$

For a subset $Q\subseteq \mathbb{Z}_{2}^{m}$, define 
\begin{equation}
    \mathcal{H}_{Q}(y_{1},y_{2},\ldots,y_{m}) := \sum\limits_{\mathbf{v}\in Q} \prod\limits_{j=1}^{m} y_{j}^{v_{j}} \in \mathbb{Z}[y_{1},y_{2},\ldots,y_{m}],
\end{equation}
where $\mathbf{v}=(v_{1},v_{2},\ldots,v_{m}).$

We next recall the following lemma from \cite{Chang_Hyun2018simplicial}.

\begin{lemma}\label{lemma_cardinality}
     Let $\mathcal{F}$ denote the set of all maximal elements of a simplicial complex $\Delta\subseteq \mathbb{Z}_{2}^{m}$. Then $$\mathcal{H}_{\Delta}(y_{1},y_{2},\ldots,y_{m}) = \sum\limits_{\emptyset \neq S \subseteq \mathcal{F}} (-1)^{|S|+1} \prod\limits_{j\in \cap S} (1+y_{j}),$$
    where $\bigcap S =  \underset{F\in S}{\bigcap}\Supp(F).$ In particular, we have $$|\Delta| = \sum\limits_{\emptyset \neq S \subseteq \mathcal{F}} (-1)^{|S|+1}2^{|\cap S|}.$$
\end{lemma}
Let $\mathbf{v}\in \mathbb{Z}_{2}^{m}$ and $\emptyset\neq A\subseteq [m]$. Define a Boolean function $\psi(\cdot\mid A) : \mathbb{Z}_{2}^{m} \to \mathbb{Z}_{2}$ by
\[
\psi(\mathbf{v} \mid A) := 
\begin{cases}
1, & \text{if } \Supp(\mathbf{v}) \cap A = \emptyset, \\
0, & \text{if } \Supp(\mathbf{v}) \cap A \neq \emptyset.
\end{cases}
\]
Then, by Lemma \ref{lemma_cardinality}, we have
\begin{equation}\label{boolean_relation}
     \sum\limits_{\mathbf{t} \in \Delta_{A}} (-1)^{\mathbf{v} \mathbf{t}} = 2^{|A|}\psi(\mathbf{v} \mid A).
\end{equation}
The following lemma is an immediate consequence of the earlier definitions.
\begin{lemma}\label{dual_psi}
    Let $\Delta_{A}\subseteq \mathbb{Z}_{2}^{m}$ be a simplicial complex. Then $\mathbf{v} \in \Delta_{A}^{\perp}$ if and only if $\psi(\mathbf{v}\mid A) = 1.$
\end{lemma}
We now recall the definition of minimal codes.
\begin{definition}
    Let $\mathcal{C}$ be a binary linear code of length $m$ and $\mathbf{0}\neq \mathbf{v} \in \mathcal{C}$. Then $\mathbf{v}$ is said to be \textit{minimal} if $\mathbf{w}\preceq \mathbf{v}$ implies $ \mathbf{w}=\mathbf{0}$ or $\mathbf{w} = \mathbf{v}$. If every non-zero codeword of $\mathcal{C}$ is minimal, then $\mathcal{C}$ is called a \textit{minimal code}.
\end{definition}
The following lemma, due to Ashikhmin and Barg \cite{Ashikhmin1998}, provides a sufficient condition for a linear code to be minimal.
\begin{lemma}\label{minimal_lemma}
    A binary linear code $\mathcal{C}$ is minimal if 
    \[
    \frac{\wt_{\min}}{\wt_{\max}}>\frac{1}{2},
    \]
     where $\wt_{\min}$ and $\wt_{\max}$ represent the minimum and maximum non-zero Hamming weights of $\mathcal{C}$, respectively.
\end{lemma}
\subsection{\texorpdfstring{$\mathcal{C}_{D}$}{CD}-construction and Weight Formula}
Let $D = D_{1}+2D_{2}\subseteq \mathbb{Z}_{4}^{m}$, where $D_{1},D_{2} \subseteq \mathbb{Z}_{2}^{m}$. Then define
\begin{equation}\label{C_D code}
    \mathcal{C}_{D}:= \{c_{D}(\mathbf{v}) = (\mathbf{v}\cdot\mathbf{d})_{\mathbf{d}\in D}: \mathbf{v}\in \mathbb{Z}_{4}^{m}\},
\end{equation}
where $\cdot$ denotes the standard Euclidean inner product on $\mathbb{Z}_{4}^{m}.$ Then $\mathcal{C}_{D}$ is a linear code over $\mathbb{Z}_{4}$ and $D$ is called the defining set of $\mathcal{C}_{D}.$ Considering all elements of the defining set $D$ as columns, we get a generating matrix $\mathcal{G}_D$ for $\mathcal{C}_{D}$. We must note that $\mathcal{G}_D$ is not a generator matrix in general, as the rows form a spanning set for $\mathcal{C}_{D}$; however, we can always obtain a generator matrix of $\mathcal{C}_{D}$ from $\mathcal{G}_D$ by applying elementary row operations using the units of $\mathbb{Z}_4$.

 Let $\mathbf{v}=\mathbf{b}+2\mathbf{c}\in\mathbb Z_4^m$, where $\mathbf{b},\mathbf{c}\in \mathbb{Z}_{2}^{m}$. Since $D=D_1+2D_2$, every $\mathbf{d}\in D$ can be written as $\mathbf{d}=\mathbf{d}_1+2\mathbf{d}_2$, where $\mathbf{d}_i\in D_i$ for $i=1,2$.
 We have
 \begin{eqnarray}\label{corrected}
     \wt_{L}(\mathbf{v}\cdot \mathbf{d}) &=& \wt_{L}(\mathbf{b}\cdot\mathbf{d}_{1}+2(\mathbf{c}\cdot\mathbf{d}_{1}+\mathbf{b}\cdot\mathbf{d}_{2})) \nonumber\\
     &=& \wt_{H}(\mathbf{c}\mathbf{d}_{1}+\mathbf{b}\mathbf{d}_{2}) + \wt_{H}((\mathbf{b}+\mathbf{c})\mathbf{d}_{1}+\mathbf{b}\mathbf{d}_{2}) \quad \text{provided $\mathbf{b}\cdot\mathbf{d}_{1} \in \{0,1\}$}.
 \end{eqnarray}
We note that considering $D=\Delta_{A}+2D_{2}$ or $D=\Delta_{A}^{*}+2D_{2}$ with $|A|=1$, guarantees that $\mathbf{b}\cdot\mathbf{d}_{1} \in \{0,1\}$ for every choice of $\mathbf{v}=\mathbf{b}+2\mathbf{c}\in\mathbb Z_4^m$ and $\mathbf{d}=\mathbf{d}_1+2\mathbf{d}_2 \in D$. Hence
\begin{align}\label{weight_expession}
     \wt_{L}(c_{D}(\mathbf{v})) &= \sum\limits_{\mathbf{d}_{1}\in D_{1}} \sum\limits_{\mathbf{d}_{2}\in D_{2}}\wt_{L}\left((\mathbf{b}+ 2 \mathbf{c})\cdot(\mathbf{d}_{1}+2\mathbf{d}_{2})\right) \nonumber \\ 
      &= \sum\limits_{\mathbf{d}_{1}\in D_{1}} \sum\limits_{\mathbf{d}_{2}\in D_{2}}\wt_{H}(\mathbf{c}\mathbf{d}_{1}+\mathbf{b}\mathbf{d}_{2}) + \sum\limits_{\mathbf{d}_{1}\in D_{1}} \sum\limits_{\mathbf{d}_{2}\in D_{2}}\wt_{H}((\mathbf{b}+\mathbf{c})\mathbf{d}_{1}+\mathbf{b}\mathbf{d}_{2}) \nonumber \\
      &= |D|-\frac{1}{2}\sum\limits_{\mathbf{d}_{1}\in D_{1}}\sum\limits_{\mathbf{d}_{2}\in D_{2}}\left(1+(-1)^{\mathbf{c}\mathbf{d}_{1}+\mathbf{b}\mathbf{d}_{2}}\right) \nonumber \\
      &\quad  + |D|-\frac{1}{2}\sum\limits_{\mathbf{d}_{1}\in D_{1}}\sum\limits_{\mathbf{d}_{2}\in D_{2}}\left(1+(-1)^{(\mathbf{b}+\mathbf{c})\mathbf{d}_{1}+\mathbf{b}\mathbf{d}_{2}}\right) \nonumber \\
      &= |D|-\frac{1}{2}\sum\limits_{\mathbf{d}_{2}\in D_{2}} (-1)^{\mathbf{b}\mathbf{d}_{2}}  \left( \sum\limits_{\mathbf{d}_{1}\in D_{1}} (-1)^{\mathbf{c}\mathbf{d}_{1}} + \sum\limits_{\mathbf{d}_{1}\in D_{1}} (-1)^{(\mathbf{b}+\mathbf{c})\mathbf{d}_{1}}  \right).
\end{align}

\begin{remark}\label{correction}
    \rm In this remark, we illustrate through examples that the main Lee weight formula used in \cite{wu2024quaternary} to calculate the Lee weight distribution of the quaternary linear codes constructed in their paper is not correct in general, and it also affects at least one of  their main theorems in Section III.
    
    The Lee weight formula as given in \cite[Eq. 3]{wu2024quaternary} is $\wt_{L}(\mathbf{u}\cdot \mathbf{t})
 = \wt_{L}\bigl(\boldsymbol{\alpha}\cdot\mathbf{t}_{1}
   + u(\boldsymbol{\beta}\cdot\mathbf{t}_{1}+\boldsymbol{\alpha}\cdot\mathbf{t}_{2})\bigr)=\wt_{H}(\boldsymbol{\beta}\mathbf{t}_{1}+\boldsymbol{\alpha}\mathbf{t}_{2})
    + \wt_{H}\bigl((\boldsymbol{\alpha}+\boldsymbol{\beta})\mathbf{t}_{1}
    + \boldsymbol{\alpha}\mathbf{t}_{2}\bigr)$. Let us consider $\mathbf{u}=\boldsymbol{\alpha}+2\boldsymbol{\beta} =(1,1,0)+2(0,0,0) \text{ and }\mathbf{t}=\mathbf{t}_1+2\mathbf{t}_2 =(1,1,0)+2(0,0,0)$.
Using the Euclidean inner product over $\mathbb{Z}_4$, we have
$\wt_L(\mathbf{u}\cdot\mathbf{t})
=\wt_L(2)=2$,
whereas $\wt_H(\boldsymbol{\beta}\mathbf{t}_1+\boldsymbol{\alpha}\mathbf{t}_2)
+\wt_H((\boldsymbol{\alpha}+\boldsymbol{\beta})\mathbf{t}_1
+\boldsymbol{\alpha}\mathbf{t}_2)
=0$. This fails because $\mathbf{\alpha}\cdot\mathbf{t}_{1} \notin \{0,1\}$.

As a consequence, the Lee weight distribution obtained in \cite[Theorem 10]{wu2024quaternary} is not correct, in general. We explain it in detail.

In \cite[Theorem 10]{wu2024quaternary}, they considered the defining set $D = \mathbb{Z}_{2}^{n} + 2\Delta^{*}$, where $n\ge 2$ and $\Delta$ is a simplicial complex generated by $\emptyset\neq A,B\subseteq[n]$. We must note that for $n\ge2$, $\mathbb{Z}_{2}^{n}$ can also be identified with a simplicial complex generated by a
single maximal element whose support has size strictly greater than one. Now, \cite[Example 13]{wu2024quaternary} illustrates \cite[Theorem 10]{wu2024quaternary} for $n = 3$, $A = \{1, 2\}$, $B = \{2, 3\}$ and the corresponding code is identified as a four-weight code with the Lee weight enumerator: $y^{80}+4x^{36}y^{44}+49x^{40}y^{40}
+8x^{44}y^{36}+2x^{52}y^{28}$, which also satisfies the Lee weight distribution obtained in \cite[Theorem 10]{wu2024quaternary}. 

However, our computations using Magma identify the corresponding quaternary code of \cite[Example 13]{wu2024quaternary} as a six-weight code with the Lee weight enumerator: 
$y^{80}+4x^{36}y^{44}+2x^{38}y^{42}
+43x^{40}y^{40}+6x^{42}y^{38}
+6x^{44}y^{36}+2x^{52}y^{28}$. This Lee weight distribution does not match with the Lee weight distribution obtained in \cite[Theorem~10]{wu2024quaternary} because the first component of their defining set in
\cite[Theorem~10]{wu2024quaternary} is $\mathbb{Z}_{2}^{n}$, a one-generator
simplicial complex whose generator has support size greater than one; consequently, $\boldsymbol{\alpha}\cdot\mathbf{t}_{1}$ may not be in $\{0,1\}$.
\end{remark}

\section{Lee Weight Distributions} \label{Sec3}
In this section, we choose two defining sets consisting of punctured two-generator simplicial complexes (Theorem \ref{maintheorem1}) and the complement of two-generator simplicial complexes (Theorem \ref{maintheorem2}) to construct quaternary linear codes. Moreover, we modified the defining set in Theorem \ref{maintheorem1} to construct projective quaternary linear codes in Theorem \ref{projective}. Furthermore, we determine their Lee weight distributions.

\begin{theorem}\label{maintheorem1}
    Let $m\geq 2$ be a positive integer and $A,B,C \subseteq [m]$ with $|A|=1$ and $B \neq C$. If $D = \Delta_{A}+ 2 (\Delta_{B,C}\setminus \Delta_{B \cap C}),$ then the code $\mathcal{C}_{D}$ defined in \eqref{C_D code} is an at most eight-weight quaternary linear code of length $2 (2^{|B|}+2^{|C|}-2^{|B \cap C|+1})$ and type $4^{1   }2^{|(B \cup C)\setminus A|}$. The corresponding Lee weight distribution is displayed in Table \ref{table1}.
\end{theorem}
\begin{proof}
    Here $D_{1}=\Delta_{A}$ and $D_{2} = \Delta_{B,C}\setminus \Delta_{B \cap C}.$ From Eq. \eqref{boolean_relation}, we have
    \[
    \sum\limits_{\mathbf{d}_{1}\in D_{1}} (-1)^{\mathbf{c}\mathbf{d}_{1}} = 2^{|A|}\psi(\mathbf{c}\mid A) =2\psi(\mathbf{c}\mid A), 
    \]
    and 
    \[
    \sum\limits_{\mathbf{d}_{1}\in D_{1}} (-1)^{(\mathbf{b}+\mathbf{c})\mathbf{d}_{1}} = 2^{|A|}\psi(\mathbf{b}+\mathbf{c}\mid A) = 2\psi(\mathbf{b}+\mathbf{c}\mid A).
    \]
Also, 
\[
    \sum\limits_{\mathbf{d}_{2}\in D_{2}} (-1)^{\mathbf{b}\mathbf{d}_{2}} = \sum\limits_{\mathbf{d}_{2}\in \Delta_{B}} (-1)^{\mathbf{b}\mathbf{d}_{2}}+\sum\limits_{\mathbf{d}_{2}\in \Delta_{C}} (-1)^{\mathbf{b}\mathbf{d}_{2}}-2\sum\limits_{\mathbf{d}_{2}\in \Delta_{B\cap C}} (-1)^{\mathbf{b}\mathbf{d}_{2}}.
\]
Then by Eq. \eqref{boolean_relation},
\[
\sum\limits_{\mathbf{d}_{2}\in D_{2}} (-1)^{\mathbf{b}\mathbf{d}_{2}} = 2^{|B|}\psi(\mathbf{b}\mid B)+2^{|C|}\psi(\mathbf{b}\mid C)-2^{|B \cap C|+1}\psi(\mathbf{b}\mid B\cap C).
\]
Now, from Eq. \eqref{weight_expession}, we have
\[
\begin{aligned}
    \wt_{L}(c_{D}(\mathbf{v})) &= |D|-\left( 2^{|B|}\psi(\mathbf{b}\mid B)+2^{|C|}\psi(\mathbf{b}\mid C)-2^{|B \cap C|+1}\psi(\mathbf{b}\mid B\cap C) \right) \\
    &\quad \times \Big( \psi(\mathbf{c}\mid A)+ \psi(\mathbf{b}+\mathbf{c}\mid A) \Big).
\end{aligned}
\]
There are $5$ possible cases which are as follows:
\begin{enumerate}
    \item[(1)] $\mathbf{b}\in \Delta_{B}^{\perp}$ and $\mathbf{b}\in \Delta_{C}^{\perp}$. Then by Lemma \ref{dual_psi}, we get
    \[
\wt_{L}(c_{D}(\mathbf{v})) = |D|-\left( 2^{|B|}+2^{|C|}-2^{|B \cap C|+1} \right) \left( \psi(\mathbf{c}\mid A)+ \psi(\mathbf{b}+\mathbf{c}\mid A) \right).
\]
\begin{itemize}
    \item If $\mathbf{c} \in \Delta_{A}^{\perp}$ and $\mathbf{b}+\mathbf{c} \in \Delta_{A}^{\perp}$, then by Lemma \ref{dual_psi}, we have 
    \[
    \wt_{L}(c_{D}(\mathbf{v})) = 0.
    \]
    In this case, 
    \begin{eqnarray*}
        \# \mathbf{b} &=& |\Delta_{A}^{\perp}\cap \Delta_{B}^{\perp}\cap \Delta_{C}^{\perp}|
        \\ 
        &=& 2^{m-|A\cup B \cup C|} \\
    \# \mathbf{c} &=& |\Delta_{A}^{\perp}| \\
    &=& 2^{m-|A|}.
    \end{eqnarray*}
    
    \item If $\mathbf{c} \in \Delta_{A}^{\perp}$ and $\mathbf{b}+\mathbf{c} \notin \Delta_{A}^{\perp}$, then by Lemma \ref{dual_psi}, we have 
    \[
    \wt_{L}(c_{D}(\mathbf{v})) = 2^{|B|}+2^{|C|}-2^{|B \cap C|+1}.
    \]
    In this case, 
    \begin{eqnarray*}
        \# \mathbf{b} &=& |\Delta_{B}^{\perp} \cap \Delta_{C}^{\perp}| -|\Delta_{A}^{\perp}\cap \Delta_{B}^{\perp}\cap \Delta_{C}^{\perp}|
        \\ 
        &=& 2^{m-|B \cup C|}-2^{m-|A\cup B \cup C|} \\
    \# \mathbf{c} &=& |\Delta_{A}^{\perp}| \\
    &=& 2^{m-|A|}.
    \end{eqnarray*}

    \item If $\mathbf{c} \notin \Delta_{A}^{\perp}$ and $\mathbf{b}+\mathbf{c} \in \Delta_{A}^{\perp}$, then by Lemma \ref{dual_psi}, we have 
    \[
    \wt_{L}(c_{D}(\mathbf{v})) = 2^{|B|}+2^{|C|}-2^{|B \cap C|+1}.
    \]
    In this case, 
    \begin{eqnarray*}
        \# \mathbf{b} 
        &=& 2^{m-|B \cup C|}-2^{m-|A\cup B \cup C|} \\
    \# \mathbf{c} &=&  2^{m-|A|}.
    \end{eqnarray*}

    \item If $\mathbf{c} \notin \Delta_{A}^{\perp}$ and $\mathbf{b}+\mathbf{c} \notin \Delta_{A}^{\perp}$, then by Lemma \ref{dual_psi}, we have 
    \[
    \wt_{L}(c_{D}(\mathbf{v})) = |D| = 2(2^{|B|}+2^{|C|}-2^{|B \cap C|+1}).
    \]
    \end{itemize}

\item[(2)] $\mathbf{b}\in \Delta_{B}^{\perp}$ and $\mathbf{b}\notin \Delta_{C}^{\perp}$. Then by Lemma \ref{dual_psi}, we get
    \[
\wt_{L}(c_{D}(\mathbf{v})) = |D|-\left( 2^{|B|}-2^{|B \cap C|+1} \right) \left( \psi(\mathbf{c}\mid A)+ \psi(\mathbf{b}+\mathbf{c}\mid A) \right).
\]
\begin{itemize}
    \item If $\mathbf{c} \in \Delta_{A}^{\perp}$ and $\mathbf{b}+\mathbf{c} \in \Delta_{A}^{\perp}$, then by Lemma \ref{dual_psi}, we have 
    \[
    \wt_{L}(c_{D}(\mathbf{v})) = 2^{|C|+1}.
    \]
    In this case, 
    \begin{eqnarray*}
        \# \mathbf{b} &=& |\Delta_{A}^{\perp}\cap \Delta_{B}^{\perp}|-|\Delta_{A}^{\perp}\cap \Delta_{B}^{\perp}\cap \Delta_{C}^{\perp}|
        \\ 
        &=& 2^{m-|A\cup B|}-2^{m-|A\cup B \cup C|} \\
    \# \mathbf{c} &=& |\Delta_{A}^{\perp}| \\
    &=& 2^{m-|A|}.
    \end{eqnarray*}
    
    \item If $\mathbf{c} \in \Delta_{A}^{\perp}$ and $\mathbf{b}+\mathbf{c} \notin \Delta_{A}^{\perp}$, then by Lemma \ref{dual_psi}, we have 
    \[
    \wt_{L}(c_{D}(\mathbf{v})) = 2^{|B|}+2^{|C|+1}-2^{|B \cap C|+1}.
    \]
    In this case, 
    \begin{eqnarray*}
        \# \mathbf{b} &=& |\Delta_{B}^{\perp}|-|\Delta_{B}^{\perp} \cap \Delta_{A}^{\perp}|-|\Delta_{B}^{\perp} \cap \Delta_{C}^{\perp}| +|\Delta_{A}^{\perp}\cap \Delta_{B}^{\perp}\cap \Delta_{C}^{\perp}|
        \\ 
        &=& 2^{m-|B|}-2^{m-|A\cup B|}-2^{m-|B \cup C|}+2^{m-|A\cup B \cup C|} \\
    \# \mathbf{c} &=& |\Delta_{A}^{\perp}| \\
    &=& 2^{m-|A|}.
    \end{eqnarray*}

    \item If $\mathbf{c} \notin \Delta_{A}^{\perp}$ and $\mathbf{b}+\mathbf{c} \in \Delta_{A}^{\perp}$, then by Lemma \ref{dual_psi}, we have 
    \[
    \wt_{L}(c_{D}(\mathbf{v})) = 2^{|B|}+2^{|C|+1}-2^{|B \cap C|+1}.
    \]
    In this case, 
    \begin{eqnarray*}
        \# \mathbf{b}  
        &=& 2^{m-|B|}-2^{m-|A\cup B|}-2^{m-|B \cup C|}+2^{m-|A\cup B \cup C|} \\
    \# \mathbf{c}
    &=& 2^{m-|A|}.
    \end{eqnarray*}

    \item If $\mathbf{c} \notin \Delta_{A}^{\perp}$ and $\mathbf{b}+\mathbf{c} \notin \Delta_{A}^{\perp}$, then by Lemma \ref{dual_psi}, we have 
    \[
    \wt_{L}(c_{D}(\mathbf{v})) = |D| = 2(2^{|B|}+2^{|C|}-2^{|B \cap C|+1}).
    \]    
\end{itemize}

\item[(3)] $\mathbf{b}\notin \Delta_{B}^{\perp}$ and $\mathbf{b}\in \Delta_{C}^{\perp}$. Then by Lemma \ref{dual_psi}, we get
    \[
\wt_{L}(c_{D}(\mathbf{v})) = |D|-\left( 2^{|C|}-2^{|B \cap C|+1} \right) \left( \psi(\mathbf{c}\mid A)+ \psi(\mathbf{b}+\mathbf{c}\mid A) \right).
\]
\begin{itemize}
    \item If $\mathbf{c} \in \Delta_{A}^{\perp}$ and $\mathbf{b}+\mathbf{c} \in \Delta_{A}^{\perp}$, then by Lemma \ref{dual_psi}, we have 
    \[
    \wt_{L}(c_{D}(\mathbf{v})) = 2^{|B|+1}.
    \]
    In this case, 
    \begin{eqnarray*}
        \# \mathbf{b} &=& |\Delta_{A}^{\perp}\cap \Delta_{C}^{\perp}|-|\Delta_{A}^{\perp}\cap \Delta_{B}^{\perp}\cap \Delta_{C}^{\perp}|
        \\ 
        &=& 2^{m-|A\cup C|}-2^{m-|A\cup B \cup C|} \\
    \# \mathbf{c} &=& |\Delta_{A}^{\perp}| \\
    &=& 2^{m-|A|}.
    \end{eqnarray*}
    
    \item If $\mathbf{c} \in \Delta_{A}^{\perp}$ and $\mathbf{b}+\mathbf{c} \notin \Delta_{A}^{\perp}$, then by Lemma \ref{dual_psi}, we have 
    \[
    \wt_{L}(c_{D}(\mathbf{v})) = 2^{|B|+1}+2^{|C|}-2^{|B \cap C|+1}.
    \]
    In this case, 
    \begin{eqnarray*}
        \# \mathbf{b} &=& |\Delta_{C}^{\perp}|-|\Delta_{C}^{\perp} \cap \Delta_{A}^{\perp}|-|\Delta_{C}^{\perp} \cap \Delta_{B}^{\perp}| +|\Delta_{A}^{\perp}\cap \Delta_{B}^{\perp}\cap \Delta_{C}^{\perp}|
        \\ 
        &=& 2^{m-|C|}-2^{m-|A\cup C|}-2^{m-|B \cup C|}+2^{m-|A\cup B \cup C|} \\
    \# \mathbf{c} &=& |\Delta_{A}^{\perp}| \\
    &=& 2^{m-|A|}.
    \end{eqnarray*}

    \item If $\mathbf{c} \notin \Delta_{A}^{\perp}$ and $\mathbf{b}+\mathbf{c} \in \Delta_{A}^{\perp}$, then by Lemma \ref{dual_psi}, we have 
    \[
    \wt_{L}(c_{D}(\mathbf{v})) = 2^{|B|+1}+2^{|C|}-2^{|B \cap C|+1}.
    \]
    In this case, 
    \begin{eqnarray*}
        \# \mathbf{b}  
        &=& 2^{m-|C|}-2^{m-|A\cup C|}-2^{m-|B \cup C|}+2^{m-|A\cup B \cup C|} \\
    \# \mathbf{c}
    &=& 2^{m-|A|}.
    \end{eqnarray*}

    \item If $\mathbf{c} \notin \Delta_{A}^{\perp}$ and $\mathbf{b}+\mathbf{c} \notin \Delta_{A}^{\perp}$, then by Lemma \ref{dual_psi}, we have 
    \[
    \wt_{L}(c_{D}(\mathbf{v})) = |D| = 2(2^{|B|}+2^{|C|}-2^{|B \cap C|+1}).
    \]    
\end{itemize}

\item[(4)] $\mathbf{b}\notin \Delta_{B}^{\perp},\mathbf{b}\notin \Delta_{C}^{\perp}$ and $\mathbf{b}\in\Delta_{B\cap C}^{\perp}$. Then by Lemma \ref{dual_psi}, we get
    \[
\wt_{L}(c_{D}(\mathbf{v})) = |D|-\left( -2^{|B \cap C|+1} \right) \left( \psi(\mathbf{c}\mid A)+ \psi(\mathbf{b}+\mathbf{c}\mid A) \right).
\]
\begin{itemize}
    \item If $\mathbf{c} \in \Delta_{A}^{\perp}$ and $\mathbf{b}+\mathbf{c} \in \Delta_{A}^{\perp}$, then by Lemma \ref{dual_psi}, we have 
    \[
    \wt_{L}(c_{D}(\mathbf{v})) = 2(2^{|B|}+2^{|C|}).
    \]
    In this case, 
    \begin{eqnarray*}
        \# \mathbf{b} &=& |\Delta_{A}^{\perp}\cap \Delta_{B\cap C}^{\perp}|-|\Delta_{A}^{\perp}\cap \Delta_{B}^{\perp}|-|\Delta_{A}^{\perp}\cap \Delta_{C}^{\perp}|+|\Delta_{A}^{\perp}\cap \Delta_{B}^{\perp}\cap\Delta_{C}^{\perp}|
        \\ 
        &=& 2^{m-|A\cup (B\cap C)|}-2^{m-|A\cup B|}-2^{m-|A \cup C|}+2^{m-|A\cup B \cup C|} \\
    \# \mathbf{c} &=& |\Delta_{A}^{\perp}| \\
    &=& 2^{m-|A|}.
    \end{eqnarray*}
    
    \item If $\mathbf{c} \in \Delta_{A}^{\perp}$ and $\mathbf{b}+\mathbf{c} \notin \Delta_{A}^{\perp}$, then by Lemma \ref{dual_psi}, we have 
    \[
    \wt_{L}(c_{D}(\mathbf{v})) = 2(2^{|B|}+2^{|C|}-2^{|B \cap C|}).
    \]
    In this case, 
    \[
    \begin{aligned}
        \# \mathbf{b} 
        &= |\Delta_{B\cap C}^{\perp}|-|\Delta_{B\cap C}^{\perp} \cap \Delta_{A}^{\perp}|-| \Delta_{B}^{\perp}|-|\Delta_{C}^{\perp}|+|\Delta_{A}^{\perp}\cap \Delta_{B}^{\perp}| \\
        &\quad +|\Delta_{A}^{\perp}\cap \Delta_{C}^{\perp}|+|\Delta_{B}^{\perp}\cap \Delta_{C}^{\perp}| -|\Delta_{A}^{\perp}\cap \Delta_{B}^{\perp}\cap \Delta_{C}^{\perp}| \\
        &= 2^{m-|B\cap C|}-2^{m-|A\cup (B\cap C)|}-2^{m-|B|}-2^{m-|C|}+2^{m-|A\cup B|} \\
        &\quad +2^{m-|A \cup C|}+2^{m-|B\cup C|}-2^{m-|A\cup B \cup C|} \\
        \# \mathbf{c} &= |\Delta_{A}^{\perp}| \\
        &= 2^{m-|A|}. 
    \end{aligned}
    \]

    \item If $\mathbf{c} \notin \Delta_{A}^{\perp}$ and $\mathbf{b}+\mathbf{c} \in \Delta_{A}^{\perp}$, then by Lemma \ref{dual_psi}, we have 
    \[
    \wt_{L}(c_{D}(\mathbf{v})) =2(2^{|B|}+2^{|C|}-2^{|B \cap C|}).
    \]
    In this case,
    \[
    \begin{aligned}
        \#\mathbf{b}
   &= 2^{m-|B\cap C|}-2^{m-|A\cup (B\cap C)|}-2^{m-|B|}-2^{m-|C|}+2^{m-|A\cup B|} \\
   &\quad +2^{m-|A \cup C|}+2^{m-|B\cup C|}-2^{m-|A\cup B \cup C|}\\
   \#\mathbf{c}
&= 2^{m-|A|}.
    \end{aligned}
    \]
    \item If $\mathbf{c} \notin \Delta_{A}^{\perp}$ and $\mathbf{b}+\mathbf{c} \notin \Delta_{A}^{\perp}$, then by Lemma \ref{dual_psi}, we have 
    \[
    \wt_{L}(c_{D}(\mathbf{v})) = |D| = 2(2^{|B|}+2^{|C|}-2^{|B \cap C|+1}).
    \]    
\end{itemize}

\item[(5)] $\mathbf{b}\notin \Delta_{B\cap C}^{\perp}$. Then by Lemma \ref{dual_psi}, we get 
\[
\wt_{L}(c_{D}(\mathbf{v})) = |D| = 2(2^{|B|}+2^{|C|}-2^{|B\cap C|+1}).
\]
\end{enumerate}
In order to find the frequency of the weight 
\[
|D| = 2(2^{|B|}+2^{|C|}-2^{|B\cap C|+1}),
\]
we sum the frequencies of all remaining weights and subtract this sum from $2^{2m}$. Hence, the resulting frequency is 
\[
2^{2m}+2^{2m-|A|-|A\cup(B\cap C)|}-2^{2m-|A|-|B\cap C|+1}.
\]
Since the frequency corresponding to Lee weight $0$ is $2^{2m-|A|-|A\cup B \cup C|}$, all the above frequencies must be divided by $2^{2m-|A|-|A\cup B \cup C|}$ to obtain the Lee weight distribution.  The total number of codewords is $2^{|A|+|A\cup B \cup C|}$. 

To determine the type of $\mathcal{C}_D$, let us consider the generating matrix $\mathcal{G}_D$ of $\mathcal{C}_D$. It is easy to observe that among the rows of $\mathcal{G}_D$, exactly $|A|$ rows contain an entry equal to $1$ or $3$, and exactly $|(B\cup C)\setminus A|$ nonzero rows have all entries in $\{0,2\}$. In order to obtain a generator matrix of $\mathcal{C}_D$ from $\mathcal{G}_D$, we can apply elementary row operations on $\mathcal{G}_D$; this yields 3 possibilities: (i) $k_{1}=|A|+s$ and $k_{2}=|(B \cup C)\setminus A|-s$, for some integer $s,$
(ii) $k_{1}=|A|$,
(iii) $k_{1}=|A|-s_1$, $k_{2}=|(B \cup C)\setminus A|-s_2$, for some integers $s_1,s_2\ge 0$.
Now, comparing this with the size of $\mathcal{C}_D$, we have $k_1 = |A|$ and $k_2 = |(B\cup C)\setminus A|$; hence we conclude that $\mathcal{C}_D$ is of type $4^{1}2^{|(B \cup C)\setminus A|}$. \qed
\end{proof}

\begin{table}[H]
\centering
\resizebox{\textwidth}{!}{
\begin{tabular}{l|l}
\hline
Weight & Frequency \\
\hline
$0$ & $1$ \\
\hline
$ w_{1} :=  2^{|B|}+2^{|C|}-2^{|B\cap C|+1} $ & $2(2^{|A\cup B \cup C|-|B\cup C|}-1)$ \\
\hline
$w_{2} := 2(2^{|B|}+2^{|C|}-2^{|B\cap C|+1})$ & $2^{|A|+|A\cup B \cup C|}+2^{|A\cup B \cup C|-|A\cup (B \cap C)|}-2^{|A\cup B \cup C|-|B \cap C|+1}$\\
\hline
$w_{3} := 2^{|C|+1}$ & $2^{|A\cup B \cup C|-|A \cup B|}-1$\\
\hline
$w_{4} := 2^{|B|}+2^{|C|+1}-2^{|B\cap C|+1}$ & $2(2^{|A\cup B \cup C|-|B|}-2^{|A\cup B \cup C|-|A \cup B|}-2^{|A\cup B \cup C|-|B\cup C|}+1)$ \\
\hline
$w_{5} := 2^{|B|+1}$ & $2^{|A\cup B \cup C|-|A\cup C|}-1$ \\
\hline
$w_{6} := 2^{|B|+1}+2^{|C|}-2^{|B\cap C|+1}$ & $2(2^{|A\cup B \cup C|-|C|}-2^{|A\cup B \cup C|-|A \cup C|}-2^{|A\cup B \cup C|-|B\cup C|}+1)$ \\ 
\hline
$w_{7} := 2(2^{|B|}+2^{|C|})$ & $2^{|A\cup B\cup C|-|A \cup (B \cap C)|}-2^{|A\cup B \cup C|-|A\cup B|}-2^{|A\cup B \cup C|-|A\cup C|}+1$ \\
\hline
$w_{8} := 2(2^{|B|}+2^{|C|}-2^{|B\cap C|})$ &
$\begin{aligned}
    2(& 2^{|A\cup B \cup C|-|B \cap C|}-2^{|A\cup B \cup C|-|A\cup (B\cap C)|}-2^{|A\cup B\cup C|-|B|} \\[-1ex]
    & -2^{|A\cup B \cup C|-|C|}+2^{|A\cup B\cup C|-|A\cup B|}+2^{|A\cup B \cup C|-|A\cup C|} \\[-1ex]
    &+ 2^{|A\cup B\cup C|-|B\cup C|}-1 )
\end{aligned}$\\
\hline
\end{tabular}
}
\caption{Lee weight distribution of Theorem \ref{maintheorem1}}
\label{table1}
\end{table}

\begin{corollary}\label{corollary1_maintheorem1}
    If $C=A$ and $B\nsubseteq C$, $C\nsubseteq B$ in Theorem \ref{maintheorem1}, then the code $\mathcal{C}_{D}$ defined in \eqref{C_D code} is of length $2^{|B|+1}$, type $4^{1}2^{|B \setminus A|}$ and minimum Lee distance $2^{|B|}+2$. Its Lee weight distribution is displayed in Table \ref{table1.1}. 
\end{corollary}

\begin{table}[H]
\centering
\begin{tabular}{l|l}
\hline
Weight & Frequency \\
\hline
$0$ & $1$ \\
\hline
$2^{|B|+1}$ & $2^{|A\cup B|+|A|}+2^{|A\cup B|-|A|+1}-2^{|A\cup B|-|A\cap B|+1}-1$ \\
\hline
$2^{|B|}+2$ & $2(2^{|A\cup B|-|B|}-1)$ \\
\hline
$2(2^{|B|}+1)$ & $2(2^{|A\cup B|-|A\cap B|}-2^{|A\cup B|-|A|}-2^{|A\cup B|-|B|}+1)$\\
\hline
\end{tabular}
\caption{Lee weight distribution of Corollary \ref{corollary1_maintheorem1}}
\label{table1.1}
\end{table}

\begin{example}
     This example illustrates Corollary \ref{corollary1_maintheorem1}. Let $m=3, A=\{1\}, B=\{2,3\}$. Then the three-weight code $\mathcal{C}_{D}$ has length $8$, type $4^1 2^2$ and minimum Lee distance $6$ with Lee weight enumerator $y^{16}+2 x^{6} y^{10} + 7 x^{8}y^{8}+6x^{10}y^{6}$, as verified using Magma \cite{bosma1997magma}. Moreover, $\mathcal{C}_{D}$ has Plotkin-defect $2$. 
\end{example}


\begin{corollary} \label{corollary3_maintheorem1}
    If $B\subsetneq C$ and $A\subseteq C$ and $A$ and $B$ are not simultaneously empty in Theorem \ref{maintheorem1}, then the code $\mathcal{C}_{D}$ in \eqref{C_D code} is of length $2(2^{|C|}-2^{|B|})$, type $4^{1} 2^{|C\setminus A|}$, minimum Lee distance $2(2^{|C|}-2^{|B|})$. Its Lee weight distribution is displayed in Table \ref{table1.3}. Moreover, $\mathcal{C}_{D}$ is Plotkin-optimal.
\end{corollary}
\begin{table}[H]
\centering
\begin{tabular}{l|l}
\hline
Weight & Frequency \\
\hline
$0$ & $1$ \\
\hline
$2^{|C|+1}$ & $2^{|C|-|A\cup B|}-1$ \\
\hline
$2^{|C|+1}-2^{|B|}$ & $2(2^{|C|-|B|}-2^{|C|-|A\cup B|})$ \\
\hline
$2(2^{|C|}-2^{|B|})$ & $2^{|A|+|C|}+2^{|C|-|A\cup B|}-2^{|C|-|B|+1}$\\
\hline
\end{tabular}
\caption{Lee weight distribution of Corollary \ref{corollary3_maintheorem1}}
\label{table1.3}
\end{table}

\begin{example}
This example illustrates Corollary \ref{corollary3_maintheorem1}. Let $m=4, A=\{1\}, B=\{2,3\}$ and $C=\{1,2,3,4\}$. Then the three-weight code $\mathcal{C}_{D}$ has length $24$, type $4^1 2^3$ and minimum Lee distance $24$ with Lee weight enumerator $y^{48}+26 x^{24}y^{24} + 4 x^{28} y^{20}+ x^{32}y^{16}$, as verified using Magma. Also, it is Plotkin-optimal. 
\end{example}

In the following theorem, we modified the defining set in Theorem \ref{maintheorem1} so that the code $\mathcal{C}_D$ is projective.

\begin{theorem}\label{projective}
     Let $m\geq 2$ be a positive integer and $A,B,C \subseteq [m]$. Consider $D = \Delta_{A}^{*}+ 2 (\Delta_{B,C}\setminus \Delta_{B \cap C}),$ with $|D|>1$, $|A|=1$ and $A\cap(B\cup C) = \emptyset$. Then, $\mathcal{C}_{D}$ is a projective code of length $2^{|B|}+2^{|C|}-2^{|B\cap C|+1}$ and the corresponding Lee weight distribution is displayed in Table \ref{table3}.
\end{theorem}
\begin{proof}
    Let $\mathbf{d} = \mathbf{x}+2\mathbf{y}\in D$. Since $\mathbf{x}\neq \mathbf{0}$ and $\Supp(\mathbf{x})\cap \Supp(\mathbf{y}) = \emptyset$, it follows that each vector in $D$ contains a coordinate equal to $1$. Moreover, no vector in $D$ has a $3$. Therefore, by Lemma \ref{lem 2.2}, the code $\mathcal{C}_{D}$ is projective. 
    
     From Eq. \eqref{boolean_relation}, we have
    \[
    \sum\limits_{\mathbf{d}_{1}\in \Delta_{A}^{*}} (-1)^{\mathbf{c}\mathbf{d}_{1}} = 2^{|A|}\psi(\mathbf{c}\mid A)-1  = 2\psi(\mathbf{c}\mid A)-1,
    \]
    and 
    \[
    \sum\limits_{\mathbf{d}_{1}\in\Delta_{A}^{*}} (-1)^{(\mathbf{b}+\mathbf{c})\mathbf{d}_{1}} = 2^{|A|}\psi(\mathbf{b}+\mathbf{c}\mid A)-1 =2\psi(\mathbf{b}+\mathbf{c}\mid A)-1.
    \]
Also, 
\[
    \sum\limits_{\mathbf{d}_{2}\in D_{2}} (-1)^{\mathbf{b}\mathbf{d}_{2}} = \sum\limits_{\mathbf{d}_{2}\in \Delta_{B}} (-1)^{\mathbf{b}\mathbf{d}_{2}}+\sum\limits_{\mathbf{d}_{2}\in \Delta_{C}} (-1)^{\mathbf{b}\mathbf{d}_{2}}-2\sum\limits_{\mathbf{d}_{2}\in \Delta_{B\cap C}} (-1)^{\mathbf{b}\mathbf{d}_{2}}.
\]
Then by Eq. \eqref{boolean_relation},
\[
\sum\limits_{\mathbf{d}_{2}\in \Delta_{B,C}\setminus \Delta_{B \cap C}} (-1)^{\mathbf{b}\mathbf{d}_{2}} = 2^{|B|}\psi(\mathbf{b}\mid B)+2^{|C|}\psi(\mathbf{b}\mid C)-2^{|B \cap C|+1}\psi(\mathbf{b}\mid B\cap C).
\]
Now, from Eq. \eqref{weight_expession}, we have
\[
\begin{aligned}
    \wt_{L}(c_{D}(\mathbf{v})) &= |D|- \left( 2^{|B|}\psi(\mathbf{b}\mid B)+2^{|C|}\psi(\mathbf{b}\mid C)-2^{|B \cap C|+1}\psi(\mathbf{b}\mid B\cap C) \right) \\
    &\quad  \times \Big( \psi(\mathbf{c}\mid A)+ \psi(\mathbf{b}+\mathbf{c}\mid A) -1 \Big).
\end{aligned}
\]
There are $5$ possible cases which are as follows:
\begin{enumerate}
    \item[(1)] $\mathbf{b}\in \Delta_{B}^{\perp}$ and $\mathbf{b}\in \Delta_{C}^{\perp}$. 
\begin{itemize}
    \item If $\mathbf{c} \in \Delta_{A}^{\perp}$ and $\mathbf{b}+\mathbf{c} \in \Delta_{A}^{\perp}$, then by Lemma \ref{dual_psi}, we have 
    \[
    \wt_{L}(c_{D}(\mathbf{v})) = 0.
    \]
    In this case, 
    \begin{eqnarray*}
        \# \mathbf{b} &=&  2^{m-|A\cup B \cup C|} \\
    \# \mathbf{c}
    &=& 2^{m-|A|}.
    \end{eqnarray*}
    
    \item If $\mathbf{c} \in \Delta_{A}^{\perp}$ and $\mathbf{b}+\mathbf{c} \notin \Delta_{A}^{\perp}$, then by Lemma \ref{dual_psi}, we have 
    \[
    \wt_{L}(c_{D}(\mathbf{v})) = 2^{|B|}+2^{|C|}-2^{|B \cap C|+1}.
    \]
    In this case, 
    \begin{eqnarray*}
        \# \mathbf{b} 
        &=& 2^{m-|B \cup C|}-2^{m-|A\cup B \cup C|} \\
    \# \mathbf{c} 
    &=& 2^{m-|A|}.
    \end{eqnarray*}

    \item If $\mathbf{c} \notin \Delta_{A}^{\perp}$ and $\mathbf{b}+\mathbf{c} \in \Delta_{A}^{\perp}$, then by Lemma \ref{dual_psi}, we have 
    \[
    \wt_{L}(c_{D}(\mathbf{v})) = 2^{|B|}+2^{|C|}-2^{|B \cap C|+1}.
    \]
    In this case, 
    \begin{eqnarray*}
        \# \mathbf{b} 
        &=& 2^{m-|B \cup C|}-2^{m-|A\cup B \cup C|} \\
    \# \mathbf{c} &=&  2^{m-|A|}.
    \end{eqnarray*}

    \item If $\mathbf{c} \notin \Delta_{A}^{\perp}$ and $\mathbf{b}+\mathbf{c} \notin \Delta_{A}^{\perp}$, then by Lemma \ref{dual_psi}, we have 
    \[
    \wt_{L}(c_{D}(\mathbf{v})) = 2(2^{|B|}+2^{|C|}-2^{|B \cap C|+1}).
    \]
    There are two subcases.
    \begin{itemize}
        \item $\mathbf{b}\in \Delta_{A}^{\perp}$. In this case, 
    \begin{eqnarray*}
        \# \mathbf{b} &=& 2^{m-|A\cup B \cup C|} \\
    \# \mathbf{c}
    &=& 2^{m}-2^{m-|A|}.
    \end{eqnarray*}
    \item $\mathbf{b}\notin \Delta_{A}^{\perp}$. In this case, 
    \begin{eqnarray*}
        \# \mathbf{b} &=& 2^{m-|B\cup C|}-2^{m-|A\cup B \cup C|} \\
    \# \mathbf{c}
    &=& 2^{m}-2^{m-|A|+1}.
    \end{eqnarray*}
    \end{itemize}
    \end{itemize}

\item[(2)] $\mathbf{b}\in \Delta_{B}^{\perp}$ and $\mathbf{b}\notin \Delta_{C}^{\perp}$. Then by Lemma \ref{dual_psi}, we get
    \[
\wt_{L}(c_{D}(\mathbf{v})) =  |D|- \left( 2^{|B|}-2^{|B \cap C|+1} \right) \Big( \psi(\mathbf{c}\mid A)+ \psi(\mathbf{b}+\mathbf{c}\mid A) -1 \Big).
\]
\begin{itemize}
    \item If $\mathbf{c} \in \Delta_{A}^{\perp}$ and $\mathbf{b}+\mathbf{c} \in \Delta_{A}^{\perp}$, then by Lemma \ref{dual_psi}, we have 
    \[
    \wt_{L}(c_{D}(\mathbf{v})) = 2^{|C|}.
    \]
    In this case, 
    \begin{eqnarray*}
        \# \mathbf{b}
        &=& 2^{m-|A\cup B|}-2^{m-|A\cup B \cup C|} \\
    \# \mathbf{c}
    &=& 2^{m-|A|}.
    \end{eqnarray*}
    
    \item If $\mathbf{c} \in \Delta_{A}^{\perp}$ and $\mathbf{b}+\mathbf{c} \notin \Delta_{A}^{\perp}$, then by Lemma \ref{dual_psi}, we have 
    \[
    \wt_{L}(c_{D}(\mathbf{v})) = 2^{|B|}+2^{|C|}-2^{|B \cap C|+1}.
    \]
    In this case, 
    \begin{eqnarray*}
        \# \mathbf{b}
        &=& 2^{m-|B|}-2^{m-|A\cup B|}-2^{m-|B \cup C|}+2^{m-|A\cup B \cup C|} \\
    \# \mathbf{c}
    &=& 2^{m-|A|}.
    \end{eqnarray*}

    \item If $\mathbf{c} \notin \Delta_{A}^{\perp}$ and $\mathbf{b}+\mathbf{c} \in \Delta_{A}^{\perp}$, then by Lemma \ref{dual_psi}, we have 
    \[
    \wt_{L}(c_{D}(\mathbf{v})) = 2^{|B|}+2^{|C|}-2^{|B \cap C|+1}.
    \]
    In this case, 
    \begin{eqnarray*}
        \# \mathbf{b}  
        &=& 2^{m-|B|}-2^{m-|A\cup B|}-2^{m-|B \cup C|}+2^{m-|A\cup B \cup C|} \\
    \# \mathbf{c}
    &=& 2^{m-|A|}.
    \end{eqnarray*}

    \item If $\mathbf{c} \notin \Delta_{A}^{\perp}$ and $\mathbf{b}+\mathbf{c} \notin \Delta_{A}^{\perp}$, then by Lemma \ref{dual_psi}, we have 
    \[
    \wt_{L}(c_{D}(\mathbf{v})) =  2(2^{|B|}+2^{|C|}-2^{|B \cap C|+1})-2^{|C|}.
    \]   
    There are two subcases.
    \begin{itemize}
        \item $\mathbf{b}\in \Delta_{A}^{\perp}$. In this case, 
    \begin{eqnarray*}
        \# \mathbf{b} &=& 2^{m-|A\cup B|}-2^{m-|A\cup B \cup C|} \\
    \# \mathbf{c}
    &=& 2^{m}-2^{m-|A|}.
    \end{eqnarray*}
    \item $\mathbf{b}\notin \Delta_{A}^{\perp}$. In this case, 
    \begin{eqnarray*}
        \# \mathbf{b} &=& 2^{m-|B|}-2^{m-|A\cup B|}-2^{m-|B\cup C|}+2^{m-|A\cup B \cup C|} \\
    \# \mathbf{c}
    &=& 2^{m}-2^{m-|A|+1}.
    \end{eqnarray*}
    \end{itemize}
\end{itemize}

\item[(3)] $\mathbf{b}\notin \Delta_{B}^{\perp}$ and $\mathbf{b}\in \Delta_{C}^{\perp}$. Then by Lemma \ref{dual_psi}, we get
    \[
\wt_{L}(c_{D}(\mathbf{v})) = |D|- \left( 2^{|C|}-2^{|B \cap C|+1} \right) \Big( \psi(\mathbf{c}\mid A)+ \psi(\mathbf{b}+\mathbf{c}\mid A) -1 \Big).
\]
\begin{itemize}
    \item If $\mathbf{c} \in \Delta_{A}^{\perp}$ and $\mathbf{b}+\mathbf{c} \in \Delta_{A}^{\perp}$, then by Lemma \ref{dual_psi}, we have 
    \[
    \wt_{L}(c_{D}(\mathbf{v})) = 2^{|B|}.
    \]
    In this case, 
    \begin{eqnarray*}
        \# \mathbf{b}
        &=& 2^{m-|A\cup C|}-2^{m-|A\cup B \cup C|} \\
    \# \mathbf{c} 
    &=& 2^{m-|A|}.
    \end{eqnarray*}
    
    \item If $\mathbf{c} \in \Delta_{A}^{\perp}$ and $\mathbf{b}+\mathbf{c} \notin \Delta_{A}^{\perp}$, then by Lemma \ref{dual_psi}, we have 
    \[
    \wt_{L}(c_{D}(\mathbf{v})) = 2^{|B|}+2^{|C|}-2^{|B \cap C|+1}.
    \]
    In this case, 
    \begin{eqnarray*}
        \# \mathbf{b}
        &=& 2^{m-|C|}-2^{m-|A\cup C|}-2^{m-|B \cup C|}+2^{m-|A\cup B \cup C|} \\
    \# \mathbf{c} 
    &=& 2^{m-|A|}.
    \end{eqnarray*}

    \item If $\mathbf{c} \notin \Delta_{A}^{\perp}$ and $\mathbf{b}+\mathbf{c} \in \Delta_{A}^{\perp}$, then by Lemma \ref{dual_psi}, we have 
    \[
    \wt_{L}(c_{D}(\mathbf{v})) = 2^{|B|}+2^{|C|}-2^{|B \cap C|+1}.
    \]
    In this case, 
    \begin{eqnarray*}
        \# \mathbf{b}  
        &=& 2^{m-|C|}-2^{m-|A\cup C|}-2^{m-|B \cup C|}+2^{m-|A\cup B \cup C|} \\
    \# \mathbf{c}
    &=& 2^{m-|A|}.
    \end{eqnarray*}

    \item If $\mathbf{c} \notin \Delta_{A}^{\perp}$ and $\mathbf{b}+\mathbf{c} \notin \Delta_{A}^{\perp}$, then by Lemma \ref{dual_psi}, we have 
    \[
    \wt_{L}(c_{D}(\mathbf{v})) = 2(2^{|B|}+2^{|C|}-2^{|B \cap C|+1})-2^{|B|}.
    \]    
    There are two subcases.
    \begin{itemize}
        \item $\mathbf{b}\in \Delta_{A}^{\perp}$. In this case, 
    \begin{eqnarray*}
        \# \mathbf{b} &=& 2^{m-|A\cup C|}-2^{m-|A\cup B \cup C|} \\
    \# \mathbf{c}
    &=& 2^{m}-2^{m-|A|}.
    \end{eqnarray*}
    \item $\mathbf{b}\notin \Delta_{A}^{\perp}$. In this case, 
    \begin{eqnarray*}
        \# \mathbf{b} &=& 2^{m-|C|}-2^{m-|A\cup C|}-2^{m-|B\cup C|}+2^{m-|A\cup B \cup C|} \\
    \# \mathbf{c}
    &=& 2^{m}-2^{m-|A|+1}.
    \end{eqnarray*}
    \end{itemize}
\end{itemize}

\item[(4)] $\mathbf{b}\notin \Delta_{B}^{\perp},\mathbf{b}\notin \Delta_{C}^{\perp}$ and $\mathbf{b}\in\Delta_{B\cap C}^{\perp}$. Then by Lemma \ref{dual_psi}, we get
    \[
\wt_{L}(c_{D}(\mathbf{v})) = |D|-\left( -2^{|B \cap C|+1} \right) \Big(\psi(\mathbf{c}\mid A)+  \psi(\mathbf{b}+\mathbf{c}\mid A) -1 \Big).
\]
\begin{itemize}
    \item If $\mathbf{c} \in \Delta_{A}^{\perp}$ and $\mathbf{b}+\mathbf{c} \in \Delta_{A}^{\perp}$, then by Lemma \ref{dual_psi}, we have 
    \[
    \wt_{L}(c_{D}(\mathbf{v})) = 2^{|B|}+2^{|C|}.
    \]
    In this case, 
    \begin{eqnarray*}
        \# \mathbf{b}
        &=& 2^{m-|A\cup (B\cap C)|}-2^{m-|A\cup B|}-2^{m-|A \cup C|}+2^{m-|A\cup B \cup C|} \\
    \# \mathbf{c} 
    &=& 2^{m-|A|}.
    \end{eqnarray*}
    
    \item If $\mathbf{c} \in \Delta_{A}^{\perp}$ and $\mathbf{b}+\mathbf{c} \notin \Delta_{A}^{\perp}$, then by Lemma \ref{dual_psi}, we have 
    \[
    \wt_{L}(c_{D}(\mathbf{v})) = 2^{|B|}+2^{|C|}-2^{|B \cap C|+1}.
    \]
    In this case, 
    \[
    \begin{aligned}
        \# \mathbf{b} 
        &= 2^{m-|B\cap C|}-2^{m-|A\cup (B\cap C)|}-2^{m-|B|}-2^{m-|C|}+2^{m-|A\cup B|} \\
        &\quad +2^{m-|A \cup C|}+2^{m-|B\cup C|}-2^{m-|A\cup B \cup C|} \\
        \# \mathbf{c}
        &= 2^{m-|A|} 
    \end{aligned}
    \]

    \item If $\mathbf{c} \notin \Delta_{A}^{\perp}$ and $\mathbf{b}+\mathbf{c} \in \Delta_{A}^{\perp}$, then by Lemma \ref{dual_psi}, we have 
    \[
    \wt_{L}(c_{D}(\mathbf{v})) = 2^{|B|}+2^{|C|}-2^{|B \cap C|+1}.
    \]
    In this case,
    \[
    \begin{aligned}
        \#\mathbf{b}
   &= 2^{m-|B\cap C|}-2^{m-|A\cup (B\cap C)|}-2^{m-|B|}-2^{m-|C|}+2^{m-|A\cup B|} \\
   &\quad +2^{m-|A \cup C|}+2^{m-|B\cup C|}-2^{m-|A\cup B \cup C|}\\
   \#\mathbf{c}
&= 2^{m-|A|}.
    \end{aligned}
    \]
    \item If $\mathbf{c} \notin \Delta_{A}^{\perp}$ and $\mathbf{b}+\mathbf{c} \notin \Delta_{A}^{\perp}$, then by Lemma \ref{dual_psi}, we have 
    \[
    \wt_{L}(c_{D}(\mathbf{v})) = 2(2^{|B|}+2^{|C|}-2^{|B \cap C|+1})-2^{|B|}-2^{|C|}.
    \]    
     There are two subcases.
    \begin{itemize}
        \item $\mathbf{b}\in \Delta_{A}^{\perp}$. In this case, 
    \begin{eqnarray*}
        \# \mathbf{b} &=& 2^{m-|A\cup (B\cap C)|}-2^{m-|A\cup B|}-2^{m-|A \cup C|}+2^{m-|A\cup B \cup C|} \\
    \# \mathbf{c}
    &=& 2^{m}-2^{m-|A|}.
    \end{eqnarray*}
    \item $\mathbf{b}\notin \Delta_{A}^{\perp}$. In this case, 
        \[
    \begin{aligned}
        \#\mathbf{b}
   &= 2^{m-|B\cap C|}-2^{m-|A\cup (B\cap C)|}-2^{m-|B|}-2^{m-|C|}+2^{m-|A\cup B|} \\
   &\quad +2^{m-|A \cup C|}+2^{m-|B\cup C|}-2^{m-|A\cup B \cup C|}\\
   \#\mathbf{c}
&= 2^{m}-2^{m-|A|+1}.
    \end{aligned}
    \]
    \end{itemize}
\end{itemize}

\item[(5)] $\mathbf{b}\notin \Delta_{B\cap C}^{\perp}$. Then by Lemma \ref{dual_psi}, we get 
\[
\wt_{L}(c_{D}(\mathbf{v})) = |D| = 2^{|B|}+2^{|C|}-2^{|B\cap C|+1}.
\]
In this case, 
    \begin{eqnarray*}
        \# \mathbf{b} &=& 2^{m}-2^{m-|B\cap C|} \\
    \# \mathbf{c}
    &=& 2^{m}.
    \end{eqnarray*}
\end{enumerate}
Since the frequency corresponding to Lee weight $0$ is $2^{2m-|A|-|A\cup B \cup C|}$, all the above frequencies must be divided by $2^{2m-|A|-|A\cup B \cup C|}$ to obtain the Lee weight distribution. \qed
\end{proof}
\begin{table}[H]
\centering
\resizebox{\textwidth}{!}{
\begin{tabular}{l|l}
\hline
Weight & Frequency \\
\hline
$0$ & $1$ \\
\hline
$ w_{1} :=  2^{|B|}+2^{|C|}-2^{|B\cap C|+1} $ & $2^{|A\cup B \cup C|+|A|}-2^{|A\cup B \cup C|-|A\cup (B\cap C)|+1}$ \\
\hline
$w_{2} := 2(2^{|B|}+2^{|C|}-2^{|B\cap C|+1})$ & $1$\\
\hline
$w_{3} := 2^{|C|}$ & $2^{|A\cup B \cup C|-|A \cup B|}-1$\\
\hline
$w_{4} := 2(2^{|B|}+2^{|C|}-2^{|B\cap C|+1})-2^{|C|}$ & $ 2^{|A\cup B \cup C|-|A\cup B|}-1 $\\
\hline
$w_{5} := 2^{|B|}$ & $2^{|A\cup B \cup C|-|A\cup C|}-1$ \\
\hline
$w_{6} := 2(2^{|B|}+2^{|C|}-2^{|B\cap C|+1})-2^{|B|}$ & $2^{|A\cup B \cup C|-|A\cup C|}-1$\\ 
\hline
$w_{7} := 2^{|B|}+2^{|C|}$ & $2^{|A\cup B\cup C|-|A \cup (B \cap C)|}-2^{|A\cup B \cup C|-|A\cup B|}-2^{|A\cup B \cup C|-|A\cup C|}+1$ \\
\hline
$w_{8} := 2(2^{|B|}+2^{|C|}-2^{|B\cap C|+1})-2^{|B|}-2^{|C|}$ & $2^{|A\cup B\cup C|-|A \cup (B \cap C)|}-2^{|A\cup B \cup C|-|A\cup B|}-2^{|A\cup B \cup C|-|A\cup C|}+1$ \\
\hline
\end{tabular}
}
\caption{Lee weight distribution of Theorem \ref{projective}}
\label{table3}
\end{table}
\begin{example}
This example illustrates Theorem \ref{projective}. Let $m=6, A=\{5\}, B=\{1,2,3,4\}$ and $C=\{2,3,4,6\}$. Then the projective two-weight code $\mathcal{C}_{D}$ has length $16$, type $4^1 2^4$ and minimum Lee distance $16$ with Lee weight enumerator $y^{32}+62 x^{16}y^{16} + x^{32}$, as verified using Magma. Moreover, it is Plotkin-optimal. 
\end{example}

\begin{remark}
    In Table \ref{table3}, if $w_{i}=0$ for some $i$ and the frequency of $0$ is $2^{j}$, then the corresponding type is $4^{1}2^{|B \cup C|-j}$.
\end{remark}

\begin{corollary}\label{cor:projective2}
    If $B\subsetneq C$ and $|B|<|C|-1$ in Theorem~\ref{projective}, then the code $\mathcal{C}_{D}$ in \eqref{C_D code} is a four-weight quaternary linear code of length $2^{|C|}-2^{|B|}$, type $4^{1}2^{|C|}$ and minimum Lee distance $2^{|C|}-2^{|B|+1}.$ Its Lee weight distribution is given in Table~\ref{tab:projective_3}. 
    \end{corollary}
    \begin{table}[H]
\centering
\begin{tabular}{l|l}
\hline
Weight & Frequency \\
\hline
$0$ & $1$ \\
\hline
$2^{|C|}-2^{|B|}$ & $2^{|C|+2}-2^{|C|-|B|+1}$ \\
\hline
$2(2^{|C|}-2^{|B|})$ & $1$\\
\hline
$2^{|C|}$ & $2^{|C|-|B|}-1$ \\
\hline
$2^{|C|}-2^{|B|+1}$ & $2^{|C|-|B|}-1$ \\
\hline

\end{tabular}
\caption{Lee weight distribution of Corollary~\ref{cor:projective2}}
\label{tab:projective_3}
\end{table}
\begin{example}
    This example illustrates Corollary \ref{cor:projective2}. Let $m=7, A=\{3\}, B=\{6,7\}$ and $C=\{1,2,4,5,6,7\}$. Then the projective four-weight code $\mathcal{C}_{D}$ has length $60$, type $4^1 2^6$ and minimum Lee distance $56$ with Lee weight enumerator $y^{120}+15 x^{56}y^{64} + 224 x^{60}y^{60}+15 x^{64}y^{56}+x^{120}$, as verified using Magma. Moreover, $\mathcal{C}_{D}$ has Plotkin-defect 4, and it improves the best-known quaternary codes with parameters $(60,4^1 2^6,52)$ reported in the database \cite{aydin2022updated}. Furthermore, the dual code $(\mathcal{C}_{D})^{\perp}$ has parameters $(60,4^{53}2^6,4)$, which are new according to the database \cite{aydin2022updated}. We also note that no Plotkin-optimal codes exist for the above lengths and types, as verified from \cite{tang2025plotkin}.
\end{example}
\begin{corollary}\label{cor:plotkin-defect1}
    If $B=\emptyset$ in Corollary~\ref{cor:projective2}, the code $\mathcal{C}_{D}$ in \eqref{C_D code} is at most a four-weight quaternary linear code of length $2^{|C|}-1$, type $4^{1}2^{|C|}$ and minimum Lee distance $2^{|C|}-2.$ Moreover, $\mathcal{C}_{D}$ has a Plotkin-defect $1$. So, codes in this family are at least almost optimal.
\end{corollary}
\begin{example}
    This example illustrates Corollary \ref{cor:plotkin-defect1}. Let $m=6, A=\{5\}, B=\emptyset$ and $C=\{1,2,3,4,6\}$. Then $\mathcal{C}_{D}$ is a projective four-weight code of length $31$, type $4^1 2^5$ and minimum Lee distance $30$ with Lee weight enumerator $y^{62}+31 x^{30}y^{32} + 64 x^{31}y^{31}+31 x^{32}y^{30}+x^{62}$ as verified using Magma. Moreover, $\mathcal{C}_{D}$ has the same parameters as the best-known quaternary code $(31,4^1 2^5,30)$ reported in the database \cite{aydin2022updated}. We must note that $\mathcal{C}_{D}$ has Plotkin-defect 1 and no Plotkin-optimal code of length $31$ and type $4^{1}2^{5}$ exists, as verified from \cite{tang2025plotkin}, hence $\mathcal{C}_{D}$ is optimal. Furthermore, the dual $\mathcal{C}_{D}^{\perp}$ also has the same parameters as the best-known quaternary code $(31,4^{25}2^{5}, 4)$ reported in the database \cite{aydin2022updated}.
\end{example}

\begin{remark}\rm 
    We analyze the code parameters obtained from Theorem \ref{projective}. Note that we report a $(n,4^{k_1}2^{k_2},d_L)$ quaternary linear code (up to length 128) in Table \ref{tableprojective} only if (i) no Plotkin-optimal code of length $n$ and type $4^{k_1}2^{k_2}$ exists according to \cite[Theorems 31, 41]{tang2025plotkin} and (ii) either no code of length $n$ and type $4^{k_1}2^{k_2}$ is listed in the database \cite{aydin2022updated} (in this case, we call it new), or the best-known quaternary code of length $n$ and type $4^{k_1}2^{k_2}$ reported in the database has a smaller Lee distance than our code with the same length and type (in this case, we call it improved).

 Additionally, when no Plotkin-optimal code of length $n$ and type $4^{k_1}2^{k_2}$ exists according to \cite{tang2025plotkin}, we report the code in Table \ref{tableprojective}, even if it matches the parameters of the best-known code of length $n$ and type $4^{k_1}2^{k_2}$ listed in the database \cite{aydin2022updated}. The reason for reporting best-known (other than new or improved) codes in this case is that these codes are projective. In contrast, (i) all the quaternary codes obtained from simplicial complexes in the literature are non-projective, and (ii) no explicit information about the projectivity of the best-known quaternary codes listed in the database is available. Therefore, even when our codes have the same parameters as the currently reported best-known codes, they may still represent an improvement due to their projectivity. We note that for the best-known quaternary codes listed in the database \cite{aydin2022updated}, we verified their projectivity from the generator matrix (if available) and only reported those cases where either the best-known code is non-projective, or a generator matrix is not available.

 We emphasize that there may exist additional new codes that cannot be reported here due to the length limitation of the database \cite{aydin2022updated}.
\end{remark}

\begin{table}[H]
    \centering
    \resizebox{\textwidth}{!}{
    \begin{tabular}
    {|c|c|c|c|c|c|c|c|c|}
    \hline
         Ref.&$m$ & $A$ & $B$ & $C$ & Length & Type & $d_L$ &Remark 
         \\
         \hline
           Cor \ref{cor:projective2}&$5$ & $\{3\}$ & $\{1,4\}$ & $\{1,2,4,5\}$ & $12$ & $4^1 2^4$ & $8$ & new \\
           &&&&&&&&Plotkin-defect $=4$ \\
           \hline
           Thm \ref{projective}&$6$ & $\{5\}$ & $\{1,2,3\}$ & $\{2,4,6\}$ & $12$ & $4^1 2^5$ & $8$ & improved; $d_L^{best}=6$ \cite{aydin2022updated}\\
         &&&&&&&&Plotkin-defect $=4$ \\
           \hline
            Cor \ref{cor:projective2}&$6$ & $\{5\}$ & $\{1,2\}$ & $\{1,2,3,4,6\}$ & $28$ & $4^1 2^5$ & $24$ & new \\
            &&&&&&&&Plotkin-defect $=4$ \\
           \hline
             Cor \ref{cor:projective2}&$7$ & $\{3\}$ & $\{6,7\}$ & $\{1,2,4,5,6,7\}$ & $60$ & $4^1 2^6$ & $56$ & improved; $d_L^{best}=52$ \cite{aydin2022updated} \\
          &&&&&&&&Plotkin-defect $=4$ \\
           \hline\hline

           Cor \ref{cor:plotkin-defect1}&$4$ & $\{4\}$ & $\emptyset$ & $\{1,2,3\}$ & $7$ & $4^1 2^3$ & $6$ & best-known \& optimal; $d_L^{best}=6$ \cite{aydin2022updated}\\
         &&&&&&&&Plotkin-defect $=1$\\
           \hline
           Cor \ref{cor:projective2}&$5$ & $\{3\}$ & $\{4\}$ & $\{1,2,4,5\}$ & $14$ & $4^1 2^4$ & $12$ & best-known; $d_L^{best}=12$ \cite{aydin2022updated}\\
         &&&&&&&&Plotkin-defect $=2$\\
         \hline
         Cor \ref{cor:plotkin-defect1}&$5$ & $\{3\}$ & $\emptyset$ & $\{1,2,4,5\}$ & $15$ & $4^1 2^4$ & $14$ & best-known \& optimal; $d_L^{best}=14$ \cite{aydin2022updated} \\
         &&&&&&&&Plotkin-defect $=1$\\
         \hline
           Cor \ref{cor:plotkin-defect1}&$6$ & $\{5\}$ & $\emptyset$ & $\{1,2,3,4,6\}$ & $31$ & $4^1 2^5$ & $30$ & best-known \& optimal; $d_L^{best}=30$ \cite{aydin2022updated}\\
          &&&&&&&& Plotkin-defect $=1$\\
           \hline
    \end{tabular}
    }
    \begin{tablenotes}
            \item[] $d_{L}^{best}$: best-known Lee distance for the given length and type listed in the database \cite{aydin2022updated}
        \end{tablenotes}
    \caption{\textbf{4 new or improved and 4 best-known projective} quaternary linear codes from Theorem \ref{projective}}
    \label{tableprojective}
\end{table}

\begin{theorem}\label{maintheorem2}
    Let $m\geq 2$ be a positive integer, and $A,B,C \subseteq [m]$ such that $|A|=1$ and $\Delta_{B,C} \neq \mathbb{Z}_{2}^{m}$. If $D = \Delta_{A}+ 2 (\Delta_{B,C})^{c},$ then the code $\mathcal{C}_{D}$ defined in \eqref{C_D code} is an at most nine-weight quaternary linear code of length $2 (2^{m}-2^{|B|}-2^{|C|}+2^{|B \cap C|})$ and type $4^{1}2^{m-1}$. The corresponding Lee weight distribution is displayed in Table \ref{table4}.
\end{theorem}
\begin{proof}
     Here $D_{1}=\Delta_{A}$ and $D_{2} = (\Delta_{B,C})^{c}.$ 
\[
    \sum\limits_{\mathbf{d}_{2}\in D_{2}} (-1)^{\mathbf{b}\mathbf{d}_{2}} = \sum\limits_{\mathbf{d}_{2}\in \mathbb{Z}_{2}^{m}} (-1)^{\mathbf{b}\mathbf{d}_{2}} -\sum\limits_{\mathbf{d}_{2}\in \Delta_{B}} (-1)^{\mathbf{b}\mathbf{d}_{2}}-\sum\limits_{\mathbf{d}_{2}\in \Delta_{C}} (-1)^{\mathbf{b}\mathbf{d}_{2}}+\sum\limits_{\mathbf{d}_{2}\in \Delta_{B\cap C}} (-1)^{\mathbf{b}\mathbf{d}_{2}}.
\]
Then by Eq. \eqref{boolean_relation} and Lemma \ref{dual_psi}, we have
\[
\sum\limits_{\mathbf{d}_{2}\in D_{2}} (-1)^{\mathbf{b}\mathbf{d}_{2}} = 
\begin{cases}
2^{m}-2^{|B|}-2^{|C|}+2^{|B\cap C|}, & \text{if \;} \mathbf{b} = \mathbf{0}, \\
-2^{|B|}\psi(\mathbf{b}\mid B)-2^{|C|}\psi(\mathbf{b}\mid C)+2^{|B\cap C|}\psi(\mathbf{b}\mid B\cap C), & \text{if \;} \mathbf{b}\neq \mathbf{0}.
\end{cases}
\]
Using Eq. \eqref{weight_expession}, we get
\[
 \wt_{L}(c_{D}(\mathbf{v})) =
\begin{cases}
|D|-|D|\psi(\mathbf{c}\mid A), & \text{if \;} \mathbf{b}=\mathbf{0}, \\
\begin{aligned}
|D| - \Big(
&-2^{|B|}\psi(\mathbf{b}\mid B)
- 2^{|C|}\psi(\mathbf{b}\mid C) \\
&+ 2^{|B\cap C|}\psi(\mathbf{b}\mid B\cap C)
\Big)
(\psi(\mathbf{c}\mid A) + \psi(\mathbf{b}+\mathbf{c}\mid A)),
\end{aligned}
& \text{if \;} \mathbf{b}\neq \mathbf{0}.
\end{cases}
\]

There are $6$ possible cases which are as follows:
\begin{enumerate}
    \item[(1)] $\mathbf{b}=\mathbf{0}$. Then 
    \[
    \wt_{L}(c_{D}(\mathbf{v}))= |D|-|D|\psi(\mathbf{c}\mid A).
    \]
    \begin{itemize}
        \item If $\mathbf{c}\in \Delta_{A}^{\perp},$ then by Lemma \ref{dual_psi}, we have
         \[
    \wt_{L}(c_{D}(\mathbf{v})) = 0.
    \]
    In this case, 
    \begin{eqnarray*}
        \# \mathbf{b} &=& 1 \\
    \# \mathbf{c} &=& |\Delta_{A}^{\perp}| \\
    &=& 2^{m-|A|}.
    \end{eqnarray*}
    
    \item  If $\mathbf{c}\notin \Delta_{A}^{\perp},$ then by Lemma \ref{dual_psi}, we have
         \[
    \wt_{L}(c_{D}(\mathbf{v})) = |D|= 2(2^{m}-2^{|B|}-2^{|C|}+2^{|B \cap C|}).
    \]
    \end{itemize}

    \item[(2)]  $\mathbf{b}\neq \mathbf{0}, \mathbf{b}\in \Delta_{B}^{\perp}$ and $\mathbf{b}\in \Delta_{C}^{\perp}$. Then by Lemma \ref{dual_psi}, we get
    \[
\wt_{L}(c_{D}(\mathbf{v})) = |D|-\left( -2^{|B|}-2^{|C|}+2^{|B \cap C|} \right) \left( \psi(\mathbf{c}\mid A)+ \psi(\mathbf{b}+\mathbf{c}\mid A) \right).
\]
\begin{itemize}
    \item If $\mathbf{c} \in \Delta_{A}^{\perp}$ and $\mathbf{b}+\mathbf{c} \in \Delta_{A}^{\perp}$, then by Lemma \ref{dual_psi}, we have 
    \[
    \wt_{L}(c_{D}(\mathbf{v})) = 2^{m+1}.
    \]
    In this case, 
    \begin{eqnarray*}
        \# \mathbf{b}
        &=& 2^{m-|A\cup B \cup C|}-1 \\
    \# \mathbf{c}
    &=& 2^{m-|A|}.
    \end{eqnarray*}
    
    \item If $\mathbf{c} \in \Delta_{A}^{\perp}$ and $\mathbf{b}+\mathbf{c} \notin \Delta_{A}^{\perp}$, then by Lemma \ref{dual_psi}, we have 
    \[
    \wt_{L}(c_{D}(\mathbf{v})) = 2^{m+1}-2^{|B|}-2^{|C|}+2^{|B \cap C|}.
    \]
    In this case, 
    \begin{eqnarray*}
        \# \mathbf{b} 
        &=& 2^{m-|B \cup C|}-2^{m-|A\cup B \cup C|} \\
    \# \mathbf{c}
    &=& 2^{m-|A|}.
    \end{eqnarray*}

    \item If $\mathbf{c} \notin \Delta_{A}^{\perp}$ and $\mathbf{b}+\mathbf{c} \in \Delta_{A}^{\perp}$, then by Lemma \ref{dual_psi}, we have 
    \[
    \wt_{L}(c_{D}(\mathbf{v})) = 2^{m+1}-2^{|B|}-2^{|C|}+2^{|B \cap C|}.
    \]
    In this case, 
    \begin{eqnarray*}
        \# \mathbf{b} 
        &=& 2^{m-|B \cup C|}-2^{m-|A\cup B \cup C|} \\
    \# \mathbf{c}
    &=& 2^{m-|A|}.
    \end{eqnarray*}

    \item If $\mathbf{c} \notin \Delta_{A}^{\perp}$ and $\mathbf{b}+\mathbf{c} \notin \Delta_{A}^{\perp}$, then by Lemma \ref{dual_psi}, we have 
    \[
    \wt_{L}(c_{D}(\mathbf{v})) = |D| = 2(2^{m}-2^{|B|}-2^{|C|}+2^{|B \cap C|}).
    \]
    \end{itemize} 

    \item[(3)] $\mathbf{b}\in \Delta_{B}^{\perp}$ and $\mathbf{b}\notin \Delta_{C}^{\perp}$. Then by Lemma \ref{dual_psi}, we get
    \[
\wt_{L}(c_{D}(\mathbf{v})) = |D|-\left(-2^{|B|}+2^{|B \cap C|} \right) \left( \psi(\mathbf{c}\mid A)+ \psi(\mathbf{b}+\mathbf{c}\mid A) \right).
\]
\begin{itemize}
    \item If $\mathbf{c} \in \Delta_{A}^{\perp}$ and $\mathbf{b}+\mathbf{c} \in \Delta_{A}^{\perp}$, then by Lemma \ref{dual_psi}, we have 
    \[
    \wt_{L}(c_{D}(\mathbf{v})) = 2(2^{m}-2^{|C|}).
    \]
    In this case, 
    \begin{eqnarray*}
        \# \mathbf{b} 
        &=& 2^{m-|A\cup B|}-2^{m-|A\cup B \cup C|} \\
    \# \mathbf{c}
    &=& 2^{m-|A|}.
    \end{eqnarray*}
    
    \item If $\mathbf{c} \in \Delta_{A}^{\perp}$ and $\mathbf{b}+\mathbf{c} \notin \Delta_{A}^{\perp}$, then by Lemma \ref{dual_psi}, we have 
    \[
    \wt_{L}(c_{D}(\mathbf{v})) = 2^{m+1}-2^{|B|}-2^{|C|+1}+2^{|B \cap C|}.
    \]
    In this case, 
    \begin{eqnarray*}
        \# \mathbf{b}
        &=& 2^{m-|B|}-2^{m-|A\cup B|}-2^{m-|B \cup C|}+2^{m-|A\cup B \cup C|} \\
    \# \mathbf{c} &=& 2^{m-|A|}.
    \end{eqnarray*}

    \item If $\mathbf{c} \notin \Delta_{A}^{\perp}$ and $\mathbf{b}+\mathbf{c} \in \Delta_{A}^{\perp}$, then by Lemma \ref{dual_psi}, we have 
    \[
    \wt_{L}(c_{D}(\mathbf{v})) = 2^{m+1}-2^{|B|}-2^{|C|+1}+2^{|B \cap C|}.
    \]
    In this case, 
    \begin{eqnarray*}
        \# \mathbf{b}
        &=& 2^{m-|B|}-2^{m-|A\cup B|}-2^{m-|B \cup C|}+2^{m-|A\cup B \cup C|} \\
    \# \mathbf{c} &=& 2^{m-|A|}.
    \end{eqnarray*}

    \item If $\mathbf{c} \notin \Delta_{A}^{\perp}$ and $\mathbf{b}+\mathbf{c} \notin \Delta_{A}^{\perp}$, then by Lemma \ref{dual_psi}, we have 
    \[
    \wt_{L}(c_{D}(\mathbf{v})) = |D| = 2(2^{m}-2^{|B|}-2^{|C|}+2^{|B \cap C|}).
    \]    
\end{itemize}

\item[(4)] $\mathbf{b}\notin \Delta_{B}^{\perp}$ and $\mathbf{b}\in \Delta_{C}^{\perp}$. Then by Lemma \ref{dual_psi}, we get
    \[
\wt_{L}(c_{D}(\mathbf{v})) = |D|-\left(-2^{|C|}+2^{|B \cap C|} \right) \left( \psi(\mathbf{c}\mid A)+ \psi(\mathbf{b}+\mathbf{c}\mid A) \right).
\]
\begin{itemize}
    \item If $\mathbf{c} \in \Delta_{A}^{\perp}$ and $\mathbf{b}+\mathbf{c} \in \Delta_{A}^{\perp}$, then by Lemma \ref{dual_psi}, we have 
    \[
    \wt_{L}(c_{D}(\mathbf{v})) = 2(2^{m}-2^{|B|}).
    \]
    In this case, 
    \begin{eqnarray*}
        \# \mathbf{b}
        &=& 2^{m-|A\cup C|}-2^{m-|A\cup B \cup C|} \\
    \# \mathbf{c} &=& 2^{m-|A|}.
    \end{eqnarray*}
    
    \item If $\mathbf{c} \in \Delta_{A}^{\perp}$ and $\mathbf{b}+\mathbf{c} \notin \Delta_{A}^{\perp}$, then by Lemma \ref{dual_psi}, we have 
    \[
    \wt_{L}(c_{D}(\mathbf{v})) = 2^{m+1}-2^{|B|+1}-2^{|C|}+2^{|B \cap C|}.
    \]
    In this case, 
    \begin{eqnarray*}
        \# \mathbf{b} &=& 2^{m-|C|}-2^{m-|A\cup C|}-2^{m-|B \cup C|}+2^{m-|A\cup B \cup C|} \\
    \# \mathbf{c} &=& 2^{m-|A|}.
    \end{eqnarray*}

    \item If $\mathbf{c} \notin \Delta_{A}^{\perp}$ and $\mathbf{b}+\mathbf{c} \in \Delta_{A}^{\perp}$, then by Lemma \ref{dual_psi}, we have 
    \[
    \wt_{L}(c_{D}(\mathbf{v})) = 2^{m+1}-2^{|B|+1}-2^{|C|}+2^{|B \cap C|}.
    \]
    In this case, 
    \begin{eqnarray*}
        \# \mathbf{b} &=& 2^{m-|C|}-2^{m-|A\cup C|}-2^{m-|B \cup C|}+2^{m-|A\cup B \cup C|} \\
    \# \mathbf{c} &=& 2^{m-|A|}.
    \end{eqnarray*}

    \item If $\mathbf{c} \notin \Delta_{A}^{\perp}$ and $\mathbf{b}+\mathbf{c} \notin \Delta_{A}^{\perp}$, then by Lemma \ref{dual_psi}, we have 
    \[
    \wt_{L}(c_{D}(\mathbf{v})) = |D| = 2(2^{m}-2^{|B|}-2^{|C|}+2^{|B \cap C|}).
    \]    
\end{itemize}

\item[(5)] $\mathbf{b}\notin \Delta_{B}^{\perp},\mathbf{b}\notin \Delta_{C}^{\perp}$ and $\mathbf{b}\in\Delta_{B\cap C}^{\perp}$. Then by Lemma \ref{dual_psi}, we get
    \[
\wt_{L}(c_{D}(\mathbf{v})) = |D|-\left( 2^{|B \cap C|} \right) \left( \psi(\mathbf{c}\mid A)+ \psi(\mathbf{b}+\mathbf{c}\mid A) \right).
\]
\begin{itemize}
    \item If $\mathbf{c} \in \Delta_{A}^{\perp}$ and $\mathbf{b}+\mathbf{c} \in \Delta_{A}^{\perp}$, then by Lemma \ref{dual_psi}, we have 
    \[
    \wt_{L}(c_{D}(\mathbf{v})) = 2(2^{m}-2^{|B|}-2^{|C|}).
    \]
    In this case, 
    \begin{eqnarray*}
        \# \mathbf{b} &=& 2^{m-|A\cup (B\cap C)|}-2^{m-|A\cup B|}-2^{m-|A \cup C|}+2^{m-|A\cup B \cup C|} \\
    \# \mathbf{c} &=& 2^{m-|A|}.
    \end{eqnarray*}
    
    \item If $\mathbf{c} \in \Delta_{A}^{\perp}$ and $\mathbf{b}+\mathbf{c} \notin \Delta_{A}^{\perp}$, then by Lemma \ref{dual_psi}, we have 
    \[
    \wt_{L}(c_{D}(\mathbf{v})) = 2^{m+1}-2^{|B|+1}-2^{|C|+1}+2^{|B \cap C|}.
    \]
    In this case, 
    \[
\begin{aligned}
\#\mathbf{b}
&= 2^{m-|B\cap C|}
 - 2^{m-|A\cup(B\cap C)|}
 - 2^{m-|B|}
 - 2^{m-|C|} + 2^{m-|A\cup B|} \\
&\quad
 + 2^{m-|A\cup C|}
 + 2^{m-|B\cup C|}
 - 2^{m-|A\cup B\cup C|} \\
\#\mathbf{c}
&= 2^{m-|A|}.
\end{aligned}
\]
    \item If $\mathbf{c} \notin \Delta_{A}^{\perp}$ and $\mathbf{b}+\mathbf{c} \in \Delta_{A}^{\perp}$, then by Lemma \ref{dual_psi}, we have 
    \[
    \wt_{L}(c_{D}(\mathbf{v})) = 2^{m+1}-2^{|B|+1}-2^{|C|+1}+2^{|B \cap C|}.
    \]
    In this case, 
    \[
\begin{aligned}
\#\mathbf{b}
&= 2^{m-|B\cap C|}
 - 2^{m-|A\cup(B\cap C)|}
 - 2^{m-|B|}
 - 2^{m-|C|} + 2^{m-|A\cup B|} \\
&\quad
 + 2^{m-|A\cup C|}
 + 2^{m-|B\cup C|}
 - 2^{m-|A\cup B\cup C|} \\
\#\mathbf{c}
&= 2^{m-|A|}.
\end{aligned}
\]
    
    \item If $\mathbf{c} \notin \Delta_{A}^{\perp}$ and $\mathbf{b}+\mathbf{c} \notin \Delta_{A}^{\perp}$, then by Lemma \ref{dual_psi}, we have 
    \[
    \wt_{L}(c_{D}(\mathbf{v})) = |D| = 2(2^{m}-2^{|B|}-2^{|C|}+2^{|B \cap C|}).
    \]    
\end{itemize}

\item[(6)] $\mathbf{b}\notin \Delta_{B\cap C}^{\perp}$. Then by Lemma \ref{dual_psi}, we get 
\[
\wt_{L}(c_{D}(\mathbf{v})) = |D| = 2(2^{m}-2^{|B|}-2^{|C|}+2^{|B\cap C|}).
\]
\end{enumerate}
In order to find the frequency of the weight 
\[
|D| = 2(2^{m}-2^{|B|}-2^{|C|}+2^{|B\cap C|}),
\]
we sum the frequencies of all remaining weights and subtract this sum from $2^{2m}$. Hence, the resulting frequency is 
\[
2^{2m}+2^{2m-|A|-|A\cup(B\cap C)|}-2^{2m-|A|-|B\cap C|+1}.
\]
Since the frequency corresponding to Lee weight $0$ is $2^{m-|A|}$, all the above frequencies must be divided by $2^{m-|A|}$ to obtain the Lee weight distribution. The total number of codewords is $2^{m+|A|}$, with parameters $k_{1} = |A|=1 $ and $k_{2}=m-|A|=m-1.$ \qed
\end{proof}

\begin{table}[H]
\centering
\resizebox{\textwidth}{!}{
\begin{tabular}{l|l}
\hline
Weight & Frequency \\
\hline
$0$ & $1$ \\
\hline
$w_{1}:= 2^{m+1}$ & $2^{m-|A\cup B\cup C|}-1$ \\
\hline
$w_{2}:= 2^{m+1}-2^{|B|}-2^{|C|}+2^{|B\cap C|}$ & $2(2^{m-|B\cup C|}-2^{m-|A\cup B\cup C|})$ \\
\hline
$w_{3}:= 2(2^{m}-2^{|C|})$ & $2^{m-|A\cup B|}-2^{m-|A\cup B\cup C|}$ \\
\hline
$w_{4}:= 2^{m+1}-2^{|B|}-2^{|C|+1}+2^{|B\cap C|}$ & $2(2^{m-|B|}-2^{m-|A\cup B|}-2^{m-|B\cup C|}+2^{m-|A\cup B\cup C|})$ \\
\hline
$w_{5}:= 2(2^{m}-2^{|B|})$ & $2^{m-|A\cup C|}-2^{m-|A\cup B\cup C|}$ \\
\hline
$w_{6}:= 2^{m+1}-2^{|B|+1}-2^{|C|}+2^{|B\cap C|}$ & $2(2^{m-|C|}-2^{m-|A\cup C|}-2^{m-|B\cup C|}+2^{m-|A\cup B\cup C|})$ \\
\hline
$w_{7}:= 2(2^{m}-2^{|B|}-2^{|C|})$ & $2^{m-|A\cup(B\cap C)|}-2^{m-|A\cup B|}-2^{m-|A\cup C|}+2^{m-|A\cup B\cup C|}$ \\
\hline
$w_{8}:= 2^{m+1}-2^{|B|+1}-2^{|C|+1}+2^{|B\cap C|}$ &
$
\begin{aligned}
    2(& 2^{m-|B\cap C|}-2^{m-|A\cup(B\cap C)|}-2^{m-|B|}-2^{m-|C|}+2^{m-|A\cup B|} \\[-1ex]
    & +2^{m-|A\cup C|}+2^{m-|B\cup C|}-2^{m-|A\cup B\cup C|})
\end{aligned}$\\
\hline
$w_{9}:= 2(2^{m}-2^{|B|}-2^{|C|}+2^{|B\cap C|})$ & $2^{m+|A|}-2^{m-|B\cap C|+1}+2^{m-|A\cup (B\cap C)|}$ \\
\hline
\end{tabular}
}
\caption{Lee weight distribution of Theorem \ref{maintheorem2}}
\label{table4}
\end{table}
\begin{example}
    This example illustrates Theorem \ref{maintheorem2}. Let $m=4, A=\{4\}, B=\{3,4\}$ and $C=\{1,2,3\}$. Then the  four-weight code $\mathcal{C}_{D}$ has length $12$, type $4^1 2^3$ and minimum Lee distance $10$ with Lee weight enumerator $y^{24}+6 x^{10}y^{14} + 20 x^{12}y^{12}+ 3 x^{16}y^{8} + 2 x^{18}y^{6}$, as verified using Magma. This improves the best-known quaternary code with parameters $(12,4^1 2^3,6)$ reported in the database \cite{aydin2022updated}. We must note that $\mathcal{C}_{D}$ has Plotkin-defect 2 and no Plotkin-optimal code of length $12$ and type $4^1 2^3$ exists by \cite{tang2025plotkin}, hence $\mathcal{C}_{D}$ is at least almost optimal.
\end{example}

\begin{corollary}\label{corollary1_maintheorem2}
    If $C=A$ and $B\nsubseteq C$, $C\nsubseteq B$ in Theorem \ref{maintheorem2}, then the code $\mathcal{C}_{D}$ has length \mbox{$2(2^{m}-2^{|B|}+2^{|A\cap B|}-2)$}, type $4^{1}2^{m-1}$ and minimum Lee distance $2^{m+1}-2^{|B|+1}+2^{|A\cap B|}-4.$
     Its Lee weight distribution is displayed in Table \ref{table4.1}.
\end{corollary}

\begin{table}[H]
\centering
\begin{tabular}{l|l}
\hline
Weight & Frequency \\
\hline
$0$ & $1$ \\
\hline
$2^{m+1}$ & $2^{m-|A\cup B|}-1$ \\
\hline
$2^{m+1}-2^{|B|}+2^{|A\cap B|}-4$ & $2(2^{m-|B|}-2^{m-|A\cup B|})$ \\
\hline
$2(2^{m}-2^{|B|})$ & $2^{m-|A|}-2^{m-|A\cup B|}$ \\
\hline
$2^{m+1}-2^{|B|+1}+2^{|A\cap B|}-4$ & $2(2^{m-|A\cap B|}-2^{m-|A|}-2^{m-|B|}+2^{m-|A\cup B|})$ \\
\hline
$2(2^{m}-2^{|B|}+2^{|A\cap B|}-2)$ & $2^{m+|A|}-2^{m-|A\cap B|+1}+2^{m-|A|}$\\
\hline
\end{tabular}
\caption{Lee weight distribution of Corollary \ref{corollary1_maintheorem2}}
\label{table4.1}
\end{table}
\begin{example}
    This example illustrates Corollary \ref{corollary1_maintheorem2}.  Let $m=5, A=\{3\}, B=\{1,4\}$ and $C=\{3\}$. Then the five-weight code $\mathcal{C}_{D}$ has length $54$, type $4^1 2^4$ and minimum Lee distance $53$ with Lee weight enumerator is given by $y^{108}+24 x^{53}y^{55} + 16 x^{54}y^{54}+ 12 x^{56}y^{52}+8 x^{57}y^{51}+ 3 x^{64}y^{44}$, as verified using Magma. This code with parameters $(54,4^12^4,53)$ is new as per the database \cite{aydin2022updated}. We must note that $\mathcal{C}_{D}$ has Plotkin-defect 1 and no Plotkin-optimal code of length $54$ and type $4^1 2^4$ exists by \cite{tang2025plotkin}, hence $\mathcal{C}_{D}$ is optimal.
\end{example}

\begin{corollary}\label{Cor:3.19}
    If $A\cap B = \emptyset$ in Corollary \ref{corollary1_maintheorem2}, $\mathcal{C}_{D}$ has Plotkin-defect $1$. 
\end{corollary}
\begin{example}
    This example illustrates Corollary \ref{Cor:3.19}. Let $m=4, A=\{4\}, B=\{1,2,3\}$ and $C=\{4\}$. Then the  four-weight code $\mathcal{C}_{D}$ has length $14$, type $4^1 2^3$ and minimum Lee distance $13$ with Lee weight enumerator $y^{28}+14 x^{13}y^{15} + 8 x^{14}y^{14}+ 7 x^{16}y^{12} + 2 x^{21}y^{7}$, as verified using Magma. This improves the best-known quaternary code with parameters $(14,4^1 2^3,12)$ reported in the database \cite{aydin2022updated}. We must note that $\mathcal{C}_{D}$ has Plotkin-defect 1 and no Plotkin-optimal code of length $14$ and type $4^1 2^3$ exists by \cite{tang2025plotkin}, hence $\mathcal{C}_{D}$ is optimal.
\end{example}

\begin{remark}\rm 
We discuss new or improved quaternary linear codes obtained from Theorem \ref{maintheorem2}.

In the following Table \ref{table5}, we list all codes (up to length $128$) obtained from Theorem \ref{maintheorem2} with new or improved parameters. We emphasize that 7 of them, having Plotkin-defect 1, are also optimal, as there exist no Plotkin-optimal codes of the same length and type as verified from \cite{tang2025plotkin}.
\begin{table}[H]
    \centering
        \resizebox{\textwidth}{!}{
    \begin{tabular}
    {|c|c|c|c|c|c|c|c|c|}
    \hline
         Ref.&$m$ & $A$ & $B$ & $C$ & Length & Type & $d_{L}$ & Remark 
         \\
         \hline
          Thm \ref{maintheorem2}&$4$ & $\{4\}$ & $\{3,4\}$ & $\{1,2,3\}$ & $12$ & $4^1 2^3$ & $10$ & improved; $d_L^{best}=6$ \cite{aydin2022updated} \\
          &&&&&&&& Plotkin-defect $=2$  \\
           \hline
          Cor \ref{Cor:3.19}&$4$ & $\{4\}$ & $\{1,2,3\}$ & $\{4\}$ & $14$ & $4^1 2^3$ & $13$ & improved \& optimal; $d_L^{best}=12$ \cite{aydin2022updated} \\
          &&&&&&&& Plotkin-defect $=1$  \\
           \hline
          
          Thm \ref{maintheorem2}&$4$ & $\{4\}$ & $\{2,3\}$ & $\{1,4\}$ & $18$ & $4^1 2^3$ & $16$ & improved; $d_L^{best}=12$ \cite{aydin2022updated} \\
          &&&&&&&& Plotkin-defect $=2$ \\
          \hline
          Cor \ref{Cor:3.19}&$4$ & $\{4\}$ & $\{2,3\}$ & $\{4\}$ & $22$ & $4^1 2^3$ & $21$ & new \& optimal \\
          &&&&&&&& Plotkin-defect $=1$ \\
          \hline
           Cor \ref{Cor:3.19}&$5$ & $\{5\}$ & $\{1,2,3,4\}$ & $\{5\}$ & $30$ & $4^1 2^4$ & $29$ & improved \& optimal; $d_L^{best}=28$ \cite{aydin2022updated}  \\
          &&&&&&&& Plotkin-defect $=1$\\
          \hline
          Thm \ref{maintheorem2}&$5$ & $\{4\}$ & $\{2,3\}$ & $\{1,4\}$ & $50$ & $4^1 2^4$ & $48$ & improved; $d_L^{best}=20$ \cite{aydin2022updated} \\
          &&&&&&&& Plotkin-defect $=2$\\
          \hline
          Cor \ref{Cor:3.19}&$5$ & $\{3\}$ & $\{1,4\}$ & $\{3\}$ & $54$ & $4^1 2^4$ & $53$ & new \& optimal \\
          &&&&&&&& Plotkin-defect $=1$\\
          \hline
          Cor \ref{Cor:3.19}&$6$ & $\{6\}$ & $\{1,2,3,4,5\}$ & $\{6\}$ & $62$ & $4^1 2^5$ & $61$ & improved \& optimal; $d_L^{best}=60$ \cite{aydin2022updated}  \\
          &&&&&&&& Plotkin-defect $=1$\\
          \hline
          Cor \ref{Cor:3.19}&$6$ & $\{6\}$ & $\{1,2,3,4\}$ & $\{6\}$ & $94$ & $4^1 2^5$ & $93$ & new \& optimal  \\
          &&&&&&&& Plotkin-defect $=1$\\
          \hline
          Cor \ref{Cor:3.19}&$6$ & $\{6\}$ & $\{1,2,3\}$ & $\{6\}$ & $110$ & $4^1 2^5$ & $109$ & new \& optimal  \\
          &&&&&&&& Plotkin-defect $=1$\\
          \hline
    \end{tabular}
    }
    \begin{tablenotes}
            \item[] $d_{L}^{best}$: best-known Lee distance for the given length and type listed in the database \cite{aydin2022updated}
        \end{tablenotes}
    \caption{\textbf{10 new or improved} quaternary linear codes from Theorem \ref{maintheorem2}}
    \label{table5}
\end{table}
\end{remark}


\section{Binary Gray images and their properties}\label{Sec4}
In the following subsection, we study the linearity of $\Phi(\mathcal{C}_{D})$ to check whether we can obtain some distance-optimal few-weight or minimal binary linear code families.
\subsection{Linearity of the Gray image}\label{subseclinearity}

\begin{theorem}\label{linear_1}
    Let $D$ be as in Theorem \ref{maintheorem1}. Then $\Phi(\mathcal{C}_{D})$ is a binary linear code. Without loss of generality, assume that $|B|\leq |C| $.
    \begin{enumerate}
        \item If $B\subsetneq C$, then $\Phi(\mathcal{C}_{D})$ has length $4(2^{|C|}-2^{|B|})$, dimension $ |A\cup C|+1$ and minimum distance 
   \[
d = 
\begin{cases}
2^{|C|}-2^{|B|} & \text{if } |A\cup C| \neq |C|, \\
2(2^{|C|}-2^{|B|}) & \text{if } |A\cup C| = |C|.
\end{cases}
\]
\item If $B\nsubseteq C$ and $|B|=|C|$, then $\Phi(\mathcal{C}_{D})$ has length $8(2^{|B|}-2^{|B\cap C|})$, dimension $ |A\cup B\cup C|+1$ and minimum distance 
 \[
d = 
\begin{cases}
2(2^{|B|}-2^{|B\cap C|}) & \text{if } A \nsubseteq B \cup C, \\
2^{|B|+1} & \text{if } A \subseteq B \cup C, |B \cup C| \neq |A\cup B| \text{ or } |B \cup C| \neq |A\cup C|, \\
3\cdot 2^{|B|}-2^{|B\cap C|+1} & \text{otherwise}.
\end{cases}
\]
\item If $B\nsubseteq C$ and $|B|<|C|$, then $\Phi(\mathcal{C}_{D})$ has length $4(2^{|B|}+2^{|C|}-2^{|B\cap C|+1})$, dimension $|A\cup B\cup C|+1$ and  minimum distance 
 \[
d = 
\begin{cases}
2^{|B|+1} & \text{if } B \nsubseteq A\cup C, \\
2^{|B|}+2^{|C|}-2^{|B\cap C|+1} & \text{if } B \subseteq A\cup C, |A\cup C| \neq |B \cup C|,\\
2^{|B|+1}+2^{|C|}-2^{|B\cap C|+1} & \text{otherwise}.
\end{cases}
\]
    \end{enumerate}
\end{theorem}
\begin{proof}
    Let $D=\{\mathbf{g}_{1},\mathbf{g}_{2},\ldots,\mathbf{g}_{|D|}\} \subseteq \mathbb{Z}_{4}^{m}$. Then 
    \[
G = (\, \mathbf{g}_1^{T}\; \mathbf{g}_2^{T}\; \cdots\; \mathbf{g}_{|D|}^{T} \,)
=
\begin{pmatrix}
\mathbf{x}_1 \\
\vdots \\
\mathbf{x}_m
\end{pmatrix},
\]
is a generating matrix for $\mathcal{C}_{D}$ so that $\{\mathbf{x}_{1},\ldots, \mathbf{x}_{m}\}$ is a set of generators of $\mathcal{C}_{D}.$

Since $|A|=1$, we have $k_{1} = 1$. Without loss of generality, assume that $\mathbf{x}_{1}$ contains $1$ or $3$. Then 
\[
2(\mathbf{x}_{i}*\mathbf{x}_{j}) = \mathbf{0} \in \mathcal{C}_{D},
\]
whenever $i$ and $j$ are not both equal to $1$. Moreover,
\[
2(\mathbf{x}_{1}*\mathbf{x}_{1}) = 2 \mathbf{x}_{1} \in \mathcal{C}_{D}.
\]
Hence, by Lemma \ref{lem 2.3}, the code $\Phi(\mathcal{C}_{D})$ is linear.\qed
 \end{proof}
 
\begin{remark}\label{distanceoptimal} \rm 
  We analyze the Gray image parameters of $\mathcal{C}_D$ as given in Corollary \ref{corollary3_maintheorem1}. If $|A|= 1$, then by Theorem \ref{linear_1}, $\Phi(\mathcal{C}_{D})$ is a binary linear code with parameters $[4(2^{|C|}-2^{|B|}), |C|+1, 2(2^{|C|}-2^{|B|})]$. This family satisfies the Griesmer bound, and its parameters coincide with those in \cite[Theorem 5]{chen2025griesmer}.  
 \end{remark}

\begin{theorem}\label{projectivelinear}
      Let $D$ be as in Theorem \ref{projective}. If $|A| =  1$, then $\Phi(C_{D})$ is a binary projective linear code. Without loss of generality, assume that $|B|\leq |C| $.
    \begin{enumerate}
        \item If $B\subsetneq C$, then $\Phi(\mathcal{C}_{D})$ has length $ 2(2^{|C|}-2^{|B|})$, dimension $|C| + 2$ and minimum distance $2^{|C|}-2^{|B|+1}$. Moreover, if $|B|<|C|-1$, then the minimum distance of the dual code is $d((\Phi(\mathcal{C}_{D}))^{\perp}) = 4$.
        \item  If $B\nsubseteq C$, then $\Phi(\mathcal{C}_{D})$ has length $2( 2^{|B|}+2^{|C|}-2^{|B\cap C|+1})$, dimension $|B \cup C|+2$ and minimum distance $2^{|B|}$.
        \end{enumerate}
 \end{theorem}
 \begin{proof}
      The proof for the linearity part follows verbatim from Theorem \ref{linear_1}. To show the minimum distance of the dual code is $d((\Phi(\mathcal{C}_{D}))^{\perp}) = 4$ when $B\subsetneq C$ with $|B|<|C|-1$, we use the MacWilliams identities as follows: Clearly, $\Phi(\mathcal{C}_{D})$ is a $[2^{|C|+1}-2^{|B|+1},|C|+2,2^{|C|}-2^{|B|+1}]$ binary linear code. 
From Table \ref{tab:projective_3}, we have 
\begin{eqnarray*}
     A_{0} &=& 1,\\
     A_{2^{|C|}-2^{|B|}} &=& 2^{|C|+2}-2^{|C|-|B|+1}, \\
     A_{2^{|C|+1}-2^{|B|+1}} &=& 1,\\
    A_{2^{|C|}} &=& 2^{|C|-|B|}-1,\\
     A_{2^{|C|}-2^{|B|+1}} &=& 2^{|C|-|B|}-1,
\end{eqnarray*}
and all remaining $A_{i}$'s are zero. Applying the MacWilliams identities \eqref{eq:MacWilliams}, we obtain
\begin{eqnarray*}
    A_{0}^{\perp} &=& 1, \\
     A_{1}^{\perp} &=& 0, \\
      A_{2}^{\perp} &=& 0, \\
       A_{3}^{\perp} &=& 0, \\
        A_{4}^{\perp} &=& \frac{1}{6}\left(2^{|C|}-2^{|B|}\right)\left(1+3\times 2^{|B|}+3\times 2^{2|B|+1}+2^{2|C|+1}-3\times 2^{|C|}(1+2^{|B|+1})\right).
\end{eqnarray*}
Since $A_{4}^{\perp}>0$ due to given conditions on $B$ and $C$, $d((\Phi(\mathcal{C}_{D}))^{\perp}) = 4$. \qed
 \end{proof}
 \begin{corollary}\label{cor:grayimage}
     If $|A|=1$ and $B\subsetneq C$ with $|B|<|C|-1$ in Theorem \ref{projective}, then the dual code $(\mathcal{C}_{D})^\perp$ of the quaternary linear code $\mathcal{C}_{D}$ in \eqref{C_D code} is a quaternary linear code of length $2^{|C|}-2^{|B|}$, type $4^{2^{|C|}-2^{|B|}-1-|C|}2^{|C|}$ and minimum Lee distance $4$.
 \end{corollary}
 \begin{proof}The parameters of $(\mathcal{C}_{D})^\perp$ are straightforward due to the fact that the Lee distance of the dual code $(\mathcal{C}_{D})^\perp$ of the quaternary linear code $\mathcal{C}_{D}$ is the same as the Hamming distance of the dual code $(\Phi(\mathcal{C}_{D}))^{\perp}$ of the Gray image of the quaternary linear code $\mathcal{C}_{D}$.
 \end{proof}
 \begin{remark}\rm 
 We discuss new or improved quaternary linear codes obtained from Corollary \ref{cor:grayimage}.

In the following Table \ref{table6}, we list all codes (up to length $128$) obtained from Corollary \ref{cor:grayimage} with new or improved parameters. 
\begin{table}[H]
    \centering
    \begin{tabular}
    {|c|c|c|c|c|c|c|c|}
    \hline
         $m$ & $A$ & $B$ & $C$ & Length & Type & $d_{L}$ & Remark 
         \\
         \hline
          $5$ & $\{5\}$ & $\{1,2\}$ & $\{1,2,3,4\}$ & $12$ & $4^7 2^4$ & $4$ & new  \\
           \hline
          $6$ & $\{5\}$ & $\{1,2,3\}$ & $\{1,2,3,4,6\}$ & $24$ & $4^{18} 2^5$ & $4$ & new \\
          \hline
          $6$ & $\{5\}$ & $\{1,2\}$ & $\{1,2,3,4,6\}$ & $28$ & $4^{22} 2^5$ & $4$ & new \\
          \hline
          $7$ & $\{3\}$ & $\{1,2,5,7\}$ & $\{1,2,4,5,6,7\}$ & $48$ & $4^{41} 2^6$ & $4$ & new \\
          \hline
          $7$ & $\{3\}$ & $\{1,2,5\}$ & $\{1,2,4,5,6,7\}$ & $56$ & $4^{49} 2^6$ & $4$ & new \\
          \hline
          $7$ & $\{3\}$ & $\{6,7\}$ & $\{1,2,4,5,6,7\}$ & $60$ & $4^{53} 2^6$ & $4$ & new \\
          \hline
          $7$ & $\{3\}$ & $\emptyset$ & $\{1,2,4,5,6,7\}$ & $63$ & $4^{56} 2^6$ & $4$ & new \\
          \hline
          $8$ & $\{1\}$ & $\{2,3,4,5,6\}$ & $\{2,3,4,5,6,7,8\}$ & $96$ & $4^{88} 2^7$ & $4$ & new \\
          \hline
          $8$ & $\{1\}$ & $\{2,3,4,5\}$ & $\{2,3,4,5,6,7,8\}$ & $112$ & $4^{104} 2^7$ & $4$ & new  \\
          \hline
    \end{tabular}
    \caption{\textbf{9 new} quaternary linear codes from Corollary \ref{cor:grayimage}}
    \label{table6}
\end{table}
\end{remark}

\begin{remark}\label{distanceoptimal2} \rm 
  We analyze the Gray image parameters of $\mathcal{C}_D$ as given in Corollary \ref{cor:plotkin-defect1}. By Theorem \ref{projectivelinear}, $\Phi(\mathcal{C}_{D})$ is a binary linear code with parameters $[2^{|C|+1}-2, |C|+2, 2^{|C|}-2]$. This family is distance-optimal by using the Griesmer bound, as there exists no binary linear code with parameters $[2^{|C|+1}-2, |C|+2, 2^{|C|}-2+s]$, for any integer $s\ge 1$ as $\sum\limits_{i=0}^{|C|+1}\left\lceil\frac{2^{|C|}-1}{2^i}\right\rceil= 2^{|C|+1}-1 > 2^{|C|+1}-2.$
  We note that there exists no binary linear code with parameters $[2^{|C|+1}-2, 2^{|C|+1}-|C|-4+s, 4]$, for any $s\ge 2$, by using the Sphere-packing bound. Hence, $(\Phi(\mathcal{C}_{D}))^\perp$ is at least almost dimension-optimal. Furthermore, using the database \cite{Grassl:codetables}, we verified that $(\Phi(\mathcal{C}_{D}))^\perp$ is both dimension-optimal and distance-optimal for $2\le |C|\le 7$.
 \end{remark}
 
 \begin{theorem}\label{linear_2}
     Let $D$ be as in Theorem \ref{maintheorem2}. Then $\Phi(C_{D})$ is a binary linear code. Without loss of generality, assume that $|B|\leq |C| $.
    \begin{enumerate}
        \item If $B\subseteq C$, then $\Phi(\mathcal{C}_{D})$ has length $4 (2^{m}-2^{|C|})$, dimension $m+1$ and minimum distance $2(2^{m}-2^{|C|})$.
        \item If $B\nsubseteq C$, then $\Phi(\mathcal{C}_{D})$ has length $4 (2^{m}-2^{|B|}-2^{|C|}+2^{|B \cap C|})$, dimension $m+1$ and minimum distance 
         \[
d = 
\begin{cases}
2^{m+1}-2^{|B|+1}-2^{|C|+1}+2^{|B\cap C|} & \text{if }    B \subseteq A \cup C \text{ or } C\subseteq A\cup B, \\
 2(2^{m}-2^{|B|}-2^{|C|}) & \text{otherwise}.
\end{cases}
\]
    \end{enumerate}
 \end{theorem}
 \begin{proof}
      The proof follows verbatim from Theorem \ref{linear_1}. \qed
 \end{proof}

In the following subsection, we study the minimality of the binary linear code $\Phi(\mathcal{C}_{D})$ after establishing its linearity in Subsection \ref{subseclinearity}.

\subsection{Minimality of the linear Gray image}
 \begin{theorem}\label{minimalthm1}
     Let $\mathcal{C}_{D}$ be as in Theorem \ref{maintheorem1}.
     \begin{enumerate}
         \item[(1)] If $B\subsetneq C$ and $A\subseteq C$, then $\Phi(\mathcal{C}_{D})$ is minimal.
         \item[(2)] If $B\nsubseteq C,|B|<|C|, B \subseteq A \cup C$ and $|A\cup C| = |B\cup C|$, then $\Phi(\mathcal{C}_{D})$ is minimal.
     \end{enumerate}
 \end{theorem}
 \begin{proof}
     For part (1), by Corollary \ref{corollary3_maintheorem1}, $\wt_{\min} = 2(2^{|C|}-2^{|B|})$ and $\wt_{\max} = 2^{|C|+1}$. Hence, by Lemma \ref{minimal_lemma},
\[
\frac{\wt_{\min}}{\wt_{\max}}>\frac{1}{2} \iff  |B| \leq |C|-2.
\]

For part (2), by Theorems \ref{maintheorem1} and \ref{linear_1}, $\wt_{\min} = 2^{|B|+1}+2^{|C|}-2^{|B\cap C|+1}$ and $\wt_{\max} = 2(2^{|B|}+2^{|C|}-2^{|B\cap C|})$. Hence, by Lemma \ref{minimal_lemma}, 
\[
\frac{\wt_{\min}}{\wt_{\max}}>\frac{1}{2} \iff  |B| \geq |B\cap C|+1,
\]
which always holds.
\qed
 \end{proof}
 \begin{example}
    This example illustrates Theorem \ref{minimalthm1} (2). Let $m=3, A = \{3\}, B=\{3\} $ and $ C = \{1,2\}$. Then the three-weight code $\Phi(\mathcal{C}_{D})$ is a minimal binary linear code with parameters $[16,4,6]$ and Hamming weight enumerator $y^{16}+2x^{6}y^{10}+7x^{8}y^{8}+6x^{10}y^{6}$, as verified by Magma.
 \end{example}
 \begin{theorem}\label{minimalthm2}
     Let $\mathcal{C}_{D}$ be as in Theorem \ref{maintheorem2}, and $|A\cup B\cup C|< m$.
     \begin{enumerate}
         \item[(1)] If $B\subseteq C$, and  $|C|\leq m-2$, then $\Phi(\mathcal{C}_{D})$ is minimal.
         \item[(2)] If $B\nsubseteq C$, $B \subseteq A\cup C$ or $C\subseteq A\cup B$, and $2^{m}+2^{|B\cap C|}>2(2^{|B|}+2^{|C|})$ then $\Phi(\mathcal{C}_{D})$ is minimal.
         \item[(3)] If $B\nsubseteq C, B \nsubseteq A\cup C$, $C\nsubseteq A\cup B$, and $2^{m}>2(2^{|B|}+2^{|C|})$, then $\Phi(\mathcal{C}_{D})$ is minimal.
     \end{enumerate}
 \end{theorem}
 \begin{proof}
     For part (1), $\wt_{\min}=2(2^m-2^{|C|})$ and $\wt_{\max} = 2^{m+1}$. Therefore, by Lemma \ref{minimal_lemma}, 
     \[
\frac{\wt_{\min}}{\wt_{\max}}>\frac{1}{2} \iff  |C| \leq m-2.
\]

For part (2), by Theorems \ref{maintheorem2} and \ref{linear_2}, $\wt_{\min} = 2^{m+1}-2^{|B|+1}-2^{|C|+1}+2^{|B\cap C|}$ and $\wt_{
\max} = 2^{m+1}$. Therefore, by Lemma \ref{minimal_lemma},
\[
\frac{\wt_{\min}}{\wt_{\max}}>\frac{1}{2} \iff  2^{m}+2^{|B\cap C|}> 2(2^{|B|}+2^{|C|}).
\]

Part (3) is analogous to part (2). \qed
 \end{proof}
 \begin{example}
     This example illustrates Theorem \ref{minimalthm2} (2). Let $m=4, A = \{1\}, B=\{1\} $ and $ C = \{2\}$. Then the five-weight code $\Phi(\mathcal{C}_{D})$ is a binary minimal linear code with parameters $[52,5,25]$ and Hamming weight enumerator $y^{52}+8x^{25}y^{27}+8x^{26}y^{26}+8x^{27}y^{25}+4x^{28}y^{24}+3x^{32}y^{20}$, as verified by Magma.
 \end{example}
  \begin{example}
  This example illustrates Theorem \ref{minimalthm2} (3). Let $m=4, A = \{1\}, B=\{2\} $ and $ C = \{3\}$. Then the seven-weight code $\Phi(\mathcal{C}_{D})$ is a binary minimal linear code with parameters $[52,5,24]$ and Hamming weight enumerator $y^{52}+2x^{24}y^{28}+4x^{25}y^{27}+8x^{26}y^{26}+8x^{27}y^{25}+4x^{28}y^{24}+4x^{29}y^{23}+x^{32}y^{20}$, as verified by Magma.
 \end{example}

\section{Conclusion}\label{Sec5}
In this article, we used defining sets consisting of punctured two-generator simplicial complexes or the complement of two-generator simplicial complexes to construct quaternary linear codes. To be particular,
\begin{itemize}
    \item In Theorem \ref{maintheorem1}, we used the defining set $\Delta_{A}+ 2 (\Delta_{B,C}\setminus \Delta_{B \cap C})$  with $|A|=1$ and found one Plotkin-optimal family in Corollary \ref{corollary3_maintheorem1}.
    \item We modified the above defining set to $\Delta_{A}^{*}+ 2 (\Delta_{B,C}\setminus \Delta_{B \cap C})$ with $|D|>1$ and $|A|=1,$ $ A\cap(B\cup C) = \emptyset$ in Theorem \ref{projective} in order to construct projective quaternary linear codes. We obtained a projective family with a Plotkin-defect 1 in Corollary \ref{cor:plotkin-defect1}. We found at least 4 new or improved parameters of projective quaternary linear codes. Also, we reported 4 projective quaternary linear codes with best-known parameters, which might outperform the currently reported best-known codes due to their projectivity. These 8 codes are listed in Table \ref{tableprojective}.
    \item In Theorem \ref{maintheorem2}, we used the defining set $\Delta_{A}+ 2 (\Delta_{B,C})^{c}$ with $|A|=1$ and obtained one family with Plotkin-defect 1 in Corollary \ref{Cor:3.19}. We reported at least 10 new or improved parameters, including 7 optimal parameters, listed in Table \ref{table5}.
    \item Furthermore, in Subsection \ref{subseclinearity}, we proved that these codes have a linear Gray image. We also found a Griesmer code family in Remark \ref{distanceoptimal}, a distance-optimal and an at least almost dimension-optimal family of binary linear codes in Remark \ref{distanceoptimal2}, along with five families of minimal binary linear codes in Theorems \ref{minimalthm1} and \ref{minimalthm2}. Moreover, using the MacWilliams identities on the Gray image parameters of quaternary linear codes constructed in Corollary \ref{cor:projective2}, we found 9 new parameters from its dual code in Table \ref{table6}.
    \item Finally, we summarized our findings in Table \ref{tablelast} and presented a consolidated list of 23 new or improved parameters, including 7 optimal and 13 projective quaternary linear codes in Table \ref{tableconsolidated}. 
\end{itemize} 

We invite further research in finding quaternary projective linear codes with small Plotkin-defects arising from other suitable choices of defining sets that outperform the new or improved codes reported in Table \ref{tableconsolidated} having a Plotkin-defect of 2 or more.

\begin{table}[H]
\centering
\resizebox{\textwidth}{!}{
\begin{tabular}{|c|c|c|c|}
\hline
Ref. & $D$ & Remarks on $\mathcal{C}_D$ & Remarks on $\Phi(\mathcal{C}_D)$ \\ 
\hline
Thm \ref{maintheorem1}&$\Delta_{A}+ 2 (\Delta_{B,C}\setminus \Delta_{B \cap C})$&a Plotkin-optimal family&a Griesmer code family\\ 
&$|A|=1$&&\cite[Theorem 5]{chen2025griesmer}\\
&&& two minimal code families \\ \hline
Thm \ref{projective}&$\Delta_{A}^{*}+ 2 (\Delta_{B,C}\setminus \Delta_{B \cap C})$ &projective families& a distance-optimal family\\ 
&$|A|=1$, $A\cap(B\cup C) = \emptyset$&a family with Plotkin-defect 1&\\
&&at least 4 new or improved and & \\
&&$4$ best-known parameters&\\
&&at least 9 new parameters from $(\mathcal{C}_{D})^\perp$&an at least almost dimension-optimal\\
&&&family from $(\Phi(\mathcal{C}_{D}))^\perp$ \\
\hline
Thm \ref{maintheorem2}&$\Delta_{A}+ 2 (\Delta_{B,C})^{c}$&
 one family with Plotkin-defect 1 & three minimal code families\\
&$|A|=1$& 
at least 10 new or improved & \\
&&
including 7 optimal parameters& \\ 
\hline
\end{tabular}
}
\caption{Summary of results obtained in this article}
\label{tablelast}
\end{table}

\begin{table}[H]
\centering
\begin{minipage}{0.48\textwidth}
\centering
\begin{tabular}{|c|c|c|c|c|}
\hline
$n$ & Type & $d_L$ &$d_L^{best}$& Remark \\
\hline

$12^*$ & $4^1 2^4$ & $8$ & $\times$ &  Plotkin-defect $=4$\\
\hline
$12^*$ & $4^1 2^5$ & $8$ & $6$ & Plotkin-defect $=4$\\
\hline
$12$ & $4^1 2^3$ & $10$ & $6$ & Plotkin-defect $=2$ \\\hline
$12^*$ & $4^7 2^4$ & $4$ & $\times$
&---\\\hline
$14$ & $4^1 2^3$ & $13$ & $12$ & optimal\\
&&&&Plotkin-defect $=1$ \\\hline
$18$ & $4^1 2^3$ & $16$ & $12$ &  Plotkin-defect $=2$ \\\hline
$22$ & $4^1 2^3$ & $21$ & $\times$ & optimal\\
&&&&Plotkin-defect $=1$ \\\hline
$24^*$ & $4^{18} 2^5$ & $4$ & $\times$
& ---\\\hline
$28^*$ & $4^1 2^5$ & $24$ & $\times$ &  Plotkin-defect $=4$\\\hline
$28^*$ & $4^{22} 2^{5}$ & $4$ & $\times$
& --- \\
\hline
$30$ & $4^1 2^4$ & $29$ & $28$ & optimal\\
&&&&Plotkin-defect $=1$ 
\\\hline
$48^*$ & $4^{41} 2^6$ & $4$ & $\times$
& ---\\\hline
\end{tabular}
\end{minipage}
\hfill
\begin{minipage}{0.48\textwidth}
\centering
\begin{tabular}{|c|c|c|c|c|}
\hline
$n$ & Type & $d_L$ &$d_L^{best}$& Remark \\
\hline
$50$ & $4^1 2^4$ & $48$ & $20$ &  Plotkin-defect $=2$ \\\hline
$54$ & $4^1 2^4$ & $53$ & $\times$ & optimal\\
&&&&Plotkin-defect $=1$ \\\hline
$56^*$ & $4^{49} 2^6$ & $4$ & $\times$
& ---\\\hline
$60^*$ & $4^1 2^6$ & $56$ & $52$ & Plotkin-defect $=4$\\\hline
$60^*$ & $4^{53} 2^6$ & $4$ & $\times$
& ---\\\hline
$62$ & $4^1 2^5$ & $61$ & $60$ & optimal\\
&&&&Plotkin-defect $=1$ \\\hline
$63^*$ & $4^{56} 2^6$ & $4$ & $\times$
& ---\\\hline
$94$ & $4^1 2^5$ & $93$ & $\times$ & optimal\\
&&&&Plotkin-defect $=1$ \\\hline
$96^*$ & $4^{88} 2^7$ & $4$ & $\times$
& ---\\\hline
$110$ & $4^1 2^5$ & $109$ & $\times$ & optimal\\
&&&&Plotkin-defect $=1$ \\\hline
$112^*$ & $4^{104} 2^7$ & $4$ & $\times$
& ---\\\hline
\end{tabular}
\end{minipage}
\caption{\textbf{23 new or improved} quaternary linear codes obtained in this article}
\label{tableconsolidated}
    \begin{tablenotes}
            \item[] *: projective
            \item[] $d_{L}^{best}$: best-known Lee distance for the given length and type listed in the database \cite{aydin2022updated}
        \end{tablenotes}

\end{table}

\section*{Acknowledgements}
The first author expresses gratitude to MHRD, India, for financial support in the form of a Senior Research Fellowship at the Indian Institute of Technology Delhi. The FIST Lab (Project No. SR/FST/MS-1/2019/45) was used to perform the computations.

\section*{Declarations}
\subsection*{Conflict of Interest}
All authors declare that they have no conflict of interest.
\bibliographystyle{abbrv}
\bibliography{ref}

@article {shi2022few-weight,
    AUTHOR = {Shi, Minjia and Li, Xiaoxiao},
     TITLE = {Few-weight codes over a non-chain ring associated with
              simplicial complexes and their distance optimal {G}ray image},
   JOURNAL = {Finite Fields Appl.},
  FJOURNAL = {Finite Fields and their Applications},
    VOLUME = {80},
      YEAR = {2022},
     PAGES = {Paper No. 101994, 15},
      ISSN = {1071-5797,1090-2465},
   MRCLASS = {94B25 (94B60)},
  MRNUMBER = {4385855},
MRREVIEWER = {Wei\ Zhao},
       DOI = {10.1016/j.ffa.2022.101994},
       URL = {https://doi.org/10.1016/j.ffa.2022.101994},
}

@article {Chang_Hyun2018simplicial,
    AUTHOR = {Chang, Seunghwan and Hyun, Jong Yoon},
     TITLE = {Linear codes from simplicial complexes},
   JOURNAL = {Des. Codes Cryptogr.},
  FJOURNAL = {Designs, Codes and Cryptography. An International Journal},
    VOLUME = {86},
      YEAR = {2018},
    NUMBER = {10},
     PAGES = {2167--2181},
      ISSN = {0925-1022,1573-7586},
   MRCLASS = {94B05 (94A60 94C10)},
  MRNUMBER = {3845307},
MRREVIEWER = {Simon\ \v Spacapan},
       DOI = {10.1007/s10623-017-0442-5},
       URL = {https://doi.org/10.1007/s10623-017-0442-5},
}

@article{aydin2022updated,
  title={An Updated Database of $\mathbb{Z}_{4}$ Codes},
  author={Aydin, Nuh and Lu, Yiang and Onta, Vishad R},
  journal={arXiv preprint arXiv:2208.06832},
  year={2022},
  note={Available: \url{http://quantumcodes.info/Z4/}}
}

@article {wu2024quaternary,
    AUTHOR = {Wu, Yansheng and Li, Chao and Zhang, Lin and Xiao, Fu},
     TITLE = {Quaternary codes and their binary images},
   JOURNAL = {IEEE Trans. Inform. Theory},
  FJOURNAL = {Institute of Electrical and Electronics Engineers.
              Transactions on Information Theory},
    VOLUME = {70},
      YEAR = {2024},
    NUMBER = {7},
     PAGES = {4759--4768},
      ISSN = {0018-9448,1557-9654},
   MRCLASS = {94B05 (94B60)},
  MRNUMBER = {4762860},
       DOI = {10.1109/tit.2024.3358332},
       URL = {https://doi.org/10.1109/tit.2024.3358332},
}

@article {tang2023general,
    AUTHOR = {Tang, Hopein Christofen and Suprijanto, Djoko},
     TITLE = {A general family of {P}lotkin-optimal two-weight codes over {$\mathbb{Z}_{4}$}},
   JOURNAL = {Des. Codes Cryptogr.},
  FJOURNAL = {Designs, Codes and Cryptography. An International Journal},
    VOLUME = {91},
      YEAR = {2023},
    NUMBER = {5},
     PAGES = {1737--1750},
      ISSN = {0925-1022,1573-7586},
   MRCLASS = {94B05 (05E30)},
  MRNUMBER = {4578162},
MRREVIEWER = {Anuradha\ Sharma},
       DOI = {10.1007/s10623-022-01176-3},
       URL = {https://doi.org/10.1007/s10623-022-01176-3},
}

@article {shi2020two,
    AUTHOR = {Shi, Minjia and Xuan, Wang and Sol\'e, Patrick},
     TITLE = {Two families of two-weight codes over {$\mathbb{Z}_4$}},
   JOURNAL = {Des. Codes Cryptogr.},
  FJOURNAL = {Designs, Codes and Cryptography. An International Journal},
    VOLUME = {88},
      YEAR = {2020},
    NUMBER = {12},
     PAGES = {2493--2505},
      ISSN = {0925-1022,1573-7586},
   MRCLASS = {94B05 (94B65)},
  MRNUMBER = {4171313},
MRREVIEWER = {Yan\ Liu},
       DOI = {10.1007/s10623-020-00796-x},
       URL = {https://doi.org/10.1007/s10623-020-00796-x},
}

@book {huffman2010book,
    AUTHOR = {Huffman, W. Cary and Pless, Vera},
     TITLE = {Fundamentals of error-correcting codes},
 PUBLISHER = {Cambridge University Press, Cambridge},
      YEAR = {2003},
     PAGES = {xviii+646},
      ISBN = {0-521-78280-5},
   MRCLASS = {94Bxx (94-01)},
  MRNUMBER = {1996953},
MRREVIEWER = {H.\ F.\ Mattson, Jr.},
       DOI = {10.1017/CBO9780511807077},
       URL = {https://doi.org/10.1017/CBO9780511807077},
}

@article{hammons1994Z4,
	  title={The $\mathbb{Z}_4$-linearity of \text{Kerdock}, \text{Preparata}, \text{Goethals}, and related codes},
  author={Hammons, A Roger and Kumar, P Vijay and Calderbank, A Robert and Sloane, Neil JA and Sol{\'e}, Patrick},
  journal={IEEE Transactions on Information Theory},
  volume={40},
  number={2},
  pages={301--319},
  year={1994},
  publisher={IEEE}
}

@book {wan1997quaternary,
    AUTHOR = {Wan, Zhe Xian},
     TITLE = {Quaternary codes},
    SERIES = {Series on Applied Mathematics},
    VOLUME = {8},
 PUBLISHER = {World Scientific Publishing Co., Inc., River Edge, NJ},
      YEAR = {1997},
     PAGES = {xii+242},
      ISBN = {981-02-3274-8},
   MRCLASS = {94B05 (11T71 94-02 94B15)},
  MRNUMBER = {1609736},
MRREVIEWER = {Bram\ van Asch},
       DOI = {10.1142/9789812798121},
       URL = {https://doi.org/10.1142/9789812798121},
}

@article {shi2017consta,
    AUTHOR = {Shi, Minjia and Qian, Liqin and Sok, Lin and Aydin, Nuh and
              Sol\'e, Patrick},
     TITLE = {On constacyclic codes over {$\mathbb{Z}_{4}[u]/\langle
              u^2-1\rangle$} and their {G}ray images},
   JOURNAL = {Finite Fields Appl.},
  FJOURNAL = {Finite Fields and their Applications},
    VOLUME = {45},
      YEAR = {2017},
     PAGES = {86--95},
      ISSN = {1071-5797,1090-2465},
   MRCLASS = {94B15 (11T71)},
  MRNUMBER = {3631355},
MRREVIEWER = {Mehmet\ E.\ K\"oro\u glu},
       DOI = {10.1016/j.ffa.2016.11.016},
       URL = {https://doi.org/10.1016/j.ffa.2016.11.016},
}

@article {Meng2024generalized,
    AUTHOR = {Meng, Xiangrui and Gao, Jian and Fu, Fang-Wei},
     TITLE = {On generalized quasi-cyclic codes over {$\mathbb{Z}_4$}},
   JOURNAL = {Discrete Math.},
  FJOURNAL = {Discrete Mathematics},
    VOLUME = {347},
      YEAR = {2024},
    NUMBER = {3},
     PAGES = {Paper No. 113821, 15},
      ISSN = {0012-365X,1872-681X},
   MRCLASS = {94B15},
  MRNUMBER = {4671061},
MRREVIEWER = {Zlatko\ G.\ Varbanov},
       DOI = {10.1016/j.disc.2023.113821},
       URL = {https://doi.org/10.1016/j.disc.2023.113821},
}

@article {habibul2021z4,
    AUTHOR = {Islam, Habibul and Prakash, Om},
     TITLE = {New {$\mathbb{Z}_4$} codes from constacyclic codes over a
              non-chain ring},
   JOURNAL = {Comput. Appl. Math.},
  FJOURNAL = {Computational \& Applied Mathematics},
    VOLUME = {40},
      YEAR = {2021},
    NUMBER = {1},
     PAGES = {Paper No. 12, 17},
      ISSN = {2238-3603,1807-0302},
   MRCLASS = {94B15 (94B05 94B35 94B60)},
  MRNUMBER = {4195964},
MRREVIEWER = {Hiram\ H.\ L\'opez},
       DOI = {10.1007/s40314-020-01398-y},
       URL = {https://doi.org/10.1007/s40314-020-01398-y},
}

@article {maheshanand2017skew,
    AUTHOR = {Sharma, Amit and Bhaintwal, Maheshanand},
     TITLE = {A class of skew-cyclic codes over {$\mathbb{Z}_{4}+u\mathbb {Z}_4$} with
              derivation},
   JOURNAL = {Adv. Math. Commun.},
  FJOURNAL = {Advances in Mathematics of Communications},
    VOLUME = {12},
      YEAR = {2018},
    NUMBER = {4},
     PAGES = {723--739},
      ISSN = {1930-5346,1930-5338},
   MRCLASS = {94B15},
  MRNUMBER = {3917639},
MRREVIEWER = {Ghulam\ Mohammad},
       DOI = {10.3934/amc.2018043},
       URL = {https://doi.org/10.3934/amc.2018043},
}

@article {hyun2020infinite,
    AUTHOR = {Hyun, Jong Yoon and Lee, Jungyun and Lee, Yoonjin},
     TITLE = {Infinite families of optimal linear codes constructed from
              simplicial complexes},
   JOURNAL = {IEEE Trans. Inform. Theory},
  FJOURNAL = {Institute of Electrical and Electronics Engineers.
              Transactions on Information Theory},
    VOLUME = {66},
      YEAR = {2020},
    NUMBER = {11},
     PAGES = {6762--6773},
      ISSN = {0018-9448,1557-9654},
   MRCLASS = {94B05 (94A60 94B65)},
  MRNUMBER = {4173606},
       DOI = {10.1109/TIT.2020.2993179},
       URL = {https://doi.org/10.1109/TIT.2020.2993179},
}

@article {ding2007cyclotomic,
    AUTHOR = {Ding, Cunsheng and Niederreiter, Harald},
     TITLE = {Cyclotomic linear codes of order 3},
   JOURNAL = {IEEE Trans. Inform. Theory},
  FJOURNAL = {Institute of Electrical and Electronics Engineers.
              Transactions on Information Theory},
    VOLUME = {53},
      YEAR = {2007},
    NUMBER = {6},
     PAGES = {2274--2277},
      ISSN = {0018-9448,1557-9654},
   MRCLASS = {94B05 (05E30)},
  MRNUMBER = {2321882},
MRREVIEWER = {Marcus\ Greferath},
       DOI = {10.1109/TIT.2007.896886},
       URL = {https://doi.org/10.1109/TIT.2007.896886},
}

@article{yadav2025optimal,
  title={Optimal binary codes from $\mathcal{C}_{D}$-codes over a non-chain ring},
  author={Yadav, Ankit and Sarma, Ritumoni and Bhagat, Anuj Kumar},
  journal={arXiv preprint arXiv:2510.09057},
  year={2025}
}

@article {vidya2024nonunital,
    AUTHOR = {Sagar, Vidya and Sarma, Ritumoni},
     TITLE = {Codes over the non-unital non-commutative ring {$E$} using
              simplicial complexes},
   JOURNAL = {IEEE Trans. Inform. Theory},
  FJOURNAL = {Institute of Electrical and Electronics Engineers.
              Transactions on Information Theory},
    VOLUME = {70},
      YEAR = {2024},
    NUMBER = {5},
     PAGES = {3373--3384},
      ISSN = {0018-9448,1557-9654},
   MRCLASS = {94B05 (94B15 94B60 94B65)},
  MRNUMBER = {4740787},
MRREVIEWER = {Elif\ Segah\ Oztas},
       DOI = {10.1109/tit.2024.3374420},
       URL = {https://doi.org/10.1109/tit.2024.3374420},
}

@article {Mondal2024mixed_alphabet,
    AUTHOR = {Mondal, Nilay Kumar and Lee, Yoonjin},
     TITLE = {Optimal binary few-weight codes using a mixed alphabet ring
              and simplicial complexes},
   JOURNAL = {IEEE Trans. Inform. Theory},
  FJOURNAL = {Institute of Electrical and Electronics Engineers.
              Transactions on Information Theory},
    VOLUME = {70},
      YEAR = {2024},
    NUMBER = {7},
     PAGES = {4865--4878},
      ISSN = {0018-9448,1557-9654},
   MRCLASS = {94B05 (94B25)},
  MRNUMBER = {4762868},
       DOI = {10.1109/tit.2024.3384067},
       URL = {https://doi.org/10.1109/tit.2024.3384067},
}

@article {wu2020optimal,
    AUTHOR = {Wu, Yansheng and Zhu, Xiaomeng and Yue, Qin},
     TITLE = {Optimal few-weight codes from simplicial complexes},
   JOURNAL = {IEEE Trans. Inform. Theory},
  FJOURNAL = {Institute of Electrical and Electronics Engineers.
              Transactions on Information Theory},
    VOLUME = {66},
      YEAR = {2020},
    NUMBER = {6},
     PAGES = {3657--3663},
      ISSN = {0018-9448,1557-9654},
   MRCLASS = {94B05 (94B65)},
  MRNUMBER = {4115124},
MRREVIEWER = {Figen\ \"Oke},
       DOI = {10.1109/TIT.2019.2946840},
       URL = {https://doi.org/10.1109/TIT.2019.2946840},
}

@article {hu2024new,
    AUTHOR = {Hu, Zhao and Xu, Yunge and Li, Nian and Zeng, Xiangyong and
              Wang, Lisha and Tang, Xiaohu},
     TITLE = {New constructions of optimal linear codes from simplicial
              complexes},
   JOURNAL = {IEEE Trans. Inform. Theory},
  FJOURNAL = {Institute of Electrical and Electronics Engineers.
              Transactions on Information Theory},
    VOLUME = {70},
      YEAR = {2024},
    NUMBER = {3},
     PAGES = {1823--1835},
      ISSN = {0018-9448,1557-9654},
   MRCLASS = {94B05},
  MRNUMBER = {4709763},
       DOI = {10.1109/tit.2023.3305609},
       URL = {https://doi.org/10.1109/tit.2023.3305609},
}

@article {chen2025optimal,
    AUTHOR = {Chen, Bing and Xu, Yunge and Hu, Zhao and Li, Nian and Zeng,
              Xiangyong},
     TITLE = {Optimal linear codes with few weights from simplicial
              complexes},
   JOURNAL = {IEEE Trans. Inform. Theory},
  FJOURNAL = {Institute of Electrical and Electronics Engineers.
              Transactions on Information Theory},
    VOLUME = {71},
      YEAR = {2025},
    NUMBER = {5},
     PAGES = {3531--3543},
      ISSN = {0018-9448,1557-9654},
   MRCLASS = {94B05},
  MRNUMBER = {4897815},
       DOI = {10.1109/tit.2025.3550182},
       URL = {https://doi.org/10.1109/tit.2025.3550182},
}

@article{wu2024survey,
  title={A survey on codes from simplicial complexes},
  author={Wu, Yansheng and Li, Chao and Hyun, Jong Yoon},
  journal={arXiv preprint arXiv:2409.16310},
  year={2024}
}

@article {anuj2025subfield,
    AUTHOR = {Bhagat, Anuj Kumar and Sarma, Ritumoni and Sagar, Vidya},
     TITLE = {Subfield codes of {$C_D$}-codes over {$\mathbb{F}_2[x]/\langle
              x^3-x\rangle$}},
   JOURNAL = {Discrete Math.},
  FJOURNAL = {Discrete Mathematics},
    VOLUME = {348},
      YEAR = {2025},
    NUMBER = {1},
     PAGES = {Paper No. 114223, 22},
      ISSN = {0012-365X,1872-681X},
   MRCLASS = {94B05 (94B60)},
  MRNUMBER = {4791283},
MRREVIEWER = {Jin\ Li},
       DOI = {10.1016/j.disc.2024.114223},
       URL = {https://doi.org/10.1016/j.disc.2024.114223},
}

@article {chen2025griesmer,
    AUTHOR = {Chen, Hao},
     TITLE = {Griesmer and optimal linear codes from the affine
              {S}olomon-{S}tiffler construction},
   JOURNAL = {IEEE Trans. Inform. Theory},
  FJOURNAL = {Institute of Electrical and Electronics Engineers.
              Transactions on Information Theory},
    VOLUME = {71},
      YEAR = {2025},
    NUMBER = {9},
     PAGES = {6834--6843},
      ISSN = {0018-9448,1557-9654},
   MRCLASS = {94B05 (94B65)},
  MRNUMBER = {4949414},
MRREVIEWER = {Joel\ E.\ Iiams},
       DOI = {10.1109/tit.2025.3588778},
       URL = {https://doi.org/10.1109/tit.2025.3588778},
}

@article {wu2024two,
    AUTHOR = {Wu, Yansheng and Li, Bowen and Fan, Weibei and Xiao, Fu},
     TITLE = {Two infinite families of quaternary codes},
   JOURNAL = {IEEE Trans. Inform. Theory},
  FJOURNAL = {Institute of Electrical and Electronics Engineers.
              Transactions on Information Theory},
    VOLUME = {70},
      YEAR = {2024},
    NUMBER = {12},
     PAGES = {8723--8733},
      ISSN = {0018-9448,1557-9654},
   MRCLASS = {94B15 (94B65)},
  MRNUMBER = {4832297},
       DOI = {10.1109/tit.2024.3454959},
       URL = {https://doi.org/10.1109/tit.2024.3454959},
}

@article {wang2026trace,
    AUTHOR = {Wang, Zhexin and Li, Nian and Zeng, Xiangyong and Tang,
              Xiaohu},
     TITLE = {Trace {C}odes {O}ver {$\mathbb{Z}_{4}$} and {T}heir {L}ee {W}eight
              {D}istributions},
   JOURNAL = {IEEE Trans. Inform. Theory},
  FJOURNAL = {Institute of Electrical and Electronics Engineers.
              Transactions on Information Theory},
    VOLUME = {72},
      YEAR = {2026},
    NUMBER = {2},
     PAGES = {1051--1066},
      ISSN = {0018-9448,1557-9654},
   MRCLASS = {99-06},
  MRNUMBER = {5021061},
       DOI = {10.1109/tit.2025.3645266},
       URL = {https://doi.org/10.1109/tit.2025.3645266},
}

@article {hyun2025multivariable,
    AUTHOR = {Hyun, Jong Yoon and Jeong, Jihye and Lee, Yoonjin},
     TITLE = {Constructing optimal few weight quaternary linear codes via
              multivariable functions},
   JOURNAL = {Cryptogr. Commun.},
  FJOURNAL = {Cryptography and Communications. Discrete Structures, Boolean
              Functions and Sequences},
    VOLUME = {17},
      YEAR = {2025},
    NUMBER = {1},
     PAGES = {57--85},
      ISSN = {1936-2447,1936-2455},
   MRCLASS = {11T71 (94B05)},
  MRNUMBER = {4856637},
       DOI = {10.1007/s12095-024-00739-6},
       URL = {https://doi.org/10.1007/s12095-024-00739-6},
}

@article {shi2019trace,
    AUTHOR = {Shi, Minjia and Liu, Yan and Randriam, Hugues and Sok, Lin and
              Sol\'e, Patrick},
     TITLE = {Trace codes over {$\mathbb{Z}_4$}, and {B}oolean functions},
   JOURNAL = {Des. Codes Cryptogr.},
  FJOURNAL = {Designs, Codes and Cryptography. An International Journal},
    VOLUME = {87},
      YEAR = {2019},
    NUMBER = {6},
     PAGES = {1447--1455},
      ISSN = {0925-1022,1573-7586},
   MRCLASS = {94B05},
  MRNUMBER = {3947355},
       DOI = {10.1007/s10623-018-0542-x},
       URL = {https://doi.org/10.1007/s10623-018-0542-x},
}

@article{zhu2019new,
  title={New quaternary codes derived from posets of the disjoint union of two chains},
  author={Zhu, Xiaomeng and Wu, Yansheng and Yue, Qin},
  journal={IEEE Communications Letters},
  volume={24},
  number={1},
  pages={20--24},
  year={2019},
  publisher={IEEE}
}

@article {yuan2006secret,
    AUTHOR = {Yuan, Jin and Ding, Cunsheng},
     TITLE = {Secret sharing schemes from three classes of linear codes},
   JOURNAL = {IEEE Trans. Inform. Theory},
  FJOURNAL = {Institute of Electrical and Electronics Engineers.
              Transactions on Information Theory},
    VOLUME = {52},
      YEAR = {2006},
    NUMBER = {1},
     PAGES = {206--212},
      ISSN = {0018-9448,1557-9654},
   MRCLASS = {94A62 (94B05)},
  MRNUMBER = {2237344},
       DOI = {10.1109/TIT.2005.860412},
       URL = {https://doi.org/10.1109/TIT.2005.860412},
}

@article {shamir1979secret,
    AUTHOR = {Shamir, Adi},
     TITLE = {How to share a secret},
   JOURNAL = {Comm. ACM},
  FJOURNAL = {Communications of the Association for Computing Machinery},
    VOLUME = {22},
      YEAR = {1979},
    NUMBER = {11},
     PAGES = {612--613},
      ISSN = {0001-0782},
   MRCLASS = {94B99 (68E99)},
  MRNUMBER = {549252},
       DOI = {10.1145/359168.359176},
       URL = {https://doi.org/10.1145/359168.359176},
}

@incollection {chabanne2014twoparty,
    AUTHOR = {Chabanne, Herv\'e{} and Cohen, G\'erard and Patey, Alain},
     TITLE = {Towards secure two-party computation from the wire-tap
              channel},
 BOOKTITLE = {Information security and cryptology---{ICISC} 2013},
    SERIES = {Lecture Notes in Comput. Sci.},
    VOLUME = {8565},
     PAGES = {34--46},
 PUBLISHER = {Springer, Cham},
      YEAR = {2014},
      ISBN = {978-3-319-12160-4; 978-3-319-12159-8},
   MRCLASS = {94A60},
  MRNUMBER = {3297255},
       DOI = {10.1007/978-3-319-12160-4\_3},
       URL = {https://doi.org/10.1007/978-3-319-12160-4_3},
}

@article {Ashikhmin1998,
    AUTHOR = {Ashikhmin, A. and Barg, A.},
     TITLE = {Minimal vectors in linear codes},
   JOURNAL = {IEEE Trans. Inform. Theory},
  FJOURNAL = {Institute of Electrical and Electronics Engineers.
              Transactions on Information Theory},
    VOLUME = {44},
      YEAR = {1998},
    NUMBER = {5},
     PAGES = {2010--2017},
      ISSN = {0018-9448,1557-9654},
   MRCLASS = {94B05 (94A62)},
  MRNUMBER = {1664103},
MRREVIEWER = {Bram\ van Asch},
       DOI = {10.1109/18.705584},
       URL = {https://doi.org/10.1109/18.705584},
}

@article{bosma1997magma,
  title={The {M}agma algebra system \textnormal{I}: The user language},
  author={Bosma, Wieb and Cannon, John and Playoust, Catherine},
  journal={J. Symb. Comput.},
  volume={24},
  number={3-4},
  pages={235--265},
  year={1997},
  publisher={Elsevier}
}

@article{wu2020binary,
  title={Binary {LCD} codes and self-orthogonal codes via simplicial complexes},
  author={Wu, Yansheng and Lee, Yoonjin},
  journal={IEEE Communications Letters},
  volume={24},
  number={6},
  pages={1159--1162},
  year={2020},
  publisher={IEEE}
}

@article {Wyner1968upperbound,
    AUTHOR = {Wyner, A. D. and Graham, R. L.},
     TITLE = {An upper bound on minimum distance for a {$k$}-ary code},
   JOURNAL = {Information and Control},
  FJOURNAL = {Information and Control},
    VOLUME = {13},
      YEAR = {1968},
     PAGES = {46--52},
      ISSN = {0019-9958,1878-2981},
   MRCLASS = {94.10},
  MRNUMBER = {235901},
}

@inproceedings{carlet2000one,
  title={One-weight $\mathbb{Z}_{4}$-linear codes},
  author={Carlet, Claude},
  booktitle={Coding Theory, Cryptography and Related Areas: Proceedings of an International Conference on Coding Theory, Cryptography and Related Areas, held in Guanajuato, Mexico, in April 1998},
  pages={57--72},
  year={2000},
  organization={Springer}
}

@article {dougherty2001maximum,
    AUTHOR = {Dougherty, Steven T. and Shiromoto, Keisuke},
     TITLE = {Maximum distance codes over rings of order 4},
   JOURNAL = {IEEE Trans. Inform. Theory},
  FJOURNAL = {Institute of Electrical and Electronics Engineers.
              Transactions on Information Theory},
    VOLUME = {47},
      YEAR = {2001},
    NUMBER = {1},
     PAGES = {400--404},
      ISSN = {0018-9448,1557-9654},
   MRCLASS = {94B65},
  MRNUMBER = {1820385},
       DOI = {10.1109/18.904544},
       URL = {https://doi.org/10.1109/18.904544},
}

@article {Kiermaier2016new,
    AUTHOR = {Kiermaier, Michael and Wassermann, Alfred and Zwanzger,
              Johannes},
     TITLE = {New upper bounds on binary linear codes and a {$\mathbb{Z}_4$}-code with a better-than-linear {G}ray image},
   JOURNAL = {IEEE Trans. Inform. Theory},
  FJOURNAL = {Institute of Electrical and Electronics Engineers.
              Transactions on Information Theory},
    VOLUME = {62},
      YEAR = {2016},
    NUMBER = {12},
     PAGES = {6768--6771},
      ISSN = {0018-9448,1557-9654},
   MRCLASS = {94B05},
  MRNUMBER = {3578230},
MRREVIEWER = {K.\ T.\ Phelps},
       DOI = {10.1109/TIT.2016.2612654},
       URL = {https://doi.org/10.1109/TIT.2016.2612654},
}

@article {Kiermaier2011hyperoval,
    AUTHOR = {Kiermaier, Michael and Zwanzger, Johannes},
     TITLE = {A {$\mathbb{Z}_4$}-linear code of high minimum {L}ee distance
              derived from a hyperoval},
   JOURNAL = {Adv. Math. Commun.},
  FJOURNAL = {Advances in Mathematics of Communications},
    VOLUME = {5},
      YEAR = {2011},
    NUMBER = {2},
     PAGES = {275--286},
      ISSN = {1930-5346,1930-5338},
   MRCLASS = {94B05 (51C05)},
  MRNUMBER = {2801593},
MRREVIEWER = {Giorgio\ Faina},
       DOI = {10.3934/amc.2011.5.275},
       URL = {https://doi.org/10.3934/amc.2011.5.275},
}

@article {Kiermaier2013new,
    AUTHOR = {Kiermaier, Michael and Zwanzger, Johannes},
     TITLE = {New ring-linear codes from dualization in projective
              {H}jelmslev geometries},
   JOURNAL = {Des. Codes Cryptogr.},
  FJOURNAL = {Designs, Codes and Cryptography. An International Journal},
    VOLUME = {66},
      YEAR = {2013},
    NUMBER = {1-3},
     PAGES = {39--55},
      ISSN = {0925-1022,1573-7586},
   MRCLASS = {94B05 (51C05)},
  MRNUMBER = {3016555},
MRREVIEWER = {Edgar\ Mart\'inez-Moro},
       DOI = {10.1007/s10623-012-9650-1},
       URL = {https://doi.org/10.1007/s10623-012-9650-1},
}

@article {shi2014optimal,
    AUTHOR = {Shi, Minjia and Wang, Yu},
     TITLE = {Optimal binary codes from one-{L}ee weight codes and two-{L}ee
              weight projective codes over {$\mathbb{Z}_4$}},
   JOURNAL = {J. Syst. Sci. Complex.},
  FJOURNAL = {Journal of Systems Science \& Complexity},
    VOLUME = {27},
      YEAR = {2014},
    NUMBER = {4},
     PAGES = {795--810},
      ISSN = {1009-6124,1559-7067},
   MRCLASS = {94B25},
  MRNUMBER = {3246028},
MRREVIEWER = {Xu\ Xiang},
       DOI = {10.1007/s11424-014-2188-8},
       URL = {https://doi.org/10.1007/s11424-014-2188-8},
}

@article {tang2025plotkin,
    AUTHOR = {Tang, Hopein Christofen},
     TITLE = {The existence of {P}lotkin-optimal linear codes over {$\mathbb{Z}_4$}},
   JOURNAL = {IEEE Trans. Inform. Theory},
  FJOURNAL = {Institute of Electrical and Electronics Engineers.
              Transactions on Information Theory},
    VOLUME = {71},
      YEAR = {2025},
    NUMBER = {9},
     PAGES = {6712--6726},
      ISSN = {0018-9448,1557-9654},
   MRCLASS = {94B05},
  MRNUMBER = {4949406},
       DOI = {10.1109/tit.2025.3582936},
       URL = {https://doi.org/10.1109/tit.2025.3582936},
}

@article {tang2023new,
    AUTHOR = {Tang, Hopein Christofen and Suprijanto, Djoko},
     TITLE = {New optimal linear codes over {${\mathbb{Z}_4}$}},
   JOURNAL = {Bull. Aust. Math. Soc.},
  FJOURNAL = {Bulletin of the Australian Mathematical Society},
    VOLUME = {107},
      YEAR = {2023},
    NUMBER = {1},
     PAGES = {158--169},
      ISSN = {0004-9727,1755-1633},
   MRCLASS = {94B05 (94B65)},
  MRNUMBER = {4531700},
MRREVIEWER = {Jens\ Zumbr\"agel},
       DOI = {10.1017/S0004972722000399},
       URL = {https://doi.org/10.1017/S0004972722000399},
}

@Misc{Grassl:codetables,
  author =       "Grassl, Markus",
  title =        "{Bounds on the minimum distance of linear codes and quantum codes}",
  howpublished = "Online available at \url{http://www.codetables.de}",
  year =         "2007",
  note =         "Accessed on 2026-07-14"
}
\end{document}